\documentclass[12pt,journal,final,onecolumn]{IEEEtran}

\makeatletter

\let\proof\@undefined
\let\endproof\@undefined
\makeatother

\usepackage[dvips]{graphicx}
\usepackage{epsfig}
\usepackage{amsmath}
\usepackage{amssymb}
\usepackage{amsthm}
\usepackage{amsfonts}
\usepackage{bm}
\usepackage{color}
\usepackage[mathscr]{eucal}
\usepackage{cite}

\graphicspath{{figs/}}

\newcommand{\mc}[1]{\mathcal{#1}}

\newcommand{\mbb}[1]{\mathbb{#1}}

\newcommand{\tsf}[1]{\textsf{#1}}

\newcommand{\defeq}{\triangleq}
\newcommand{\Pp}{\mathbb{P}}

\newcommand{\ind}{1\hspace{-1.3mm}{1}}
\newcommand{\E}{\mathbb{E}}

\newcommand{\ceil}[1]{\left\lceil{#1}\right\rceil}
\newcommand{\floor}[1]{\left\lfloor{#1}\right\rfloor}
\newcommand{\abs}[1]{\lvert{#1}\rvert}
\newcommand{\card}[1]{\abs{#1}}
\newcommand{\norm}[1]{\lVert{#1}\rVert}

\newcommand{\tr}{\mathrm{tr}}

\newcommand{\hy}{\hat{y}}

\newcommand{\lambdauc}{\lambda}

\newtheorem{lemma}{Lemma}

\newtheorem{theorem}[lemma]{Theorem}

\theoremstyle{definition}
\newtheorem{egdummy}{Example}

\theoremstyle{remark}

\begin{document}

\bibliographystyle{unsrt}

\title{On Capacity Scaling in Arbitrary Wireless Networks} 

\author{Urs Niesen, Piyush Gupta, and Devavrat Shah
\thanks{U.~Niesen and D.~Shah are with the Laboratory of
Information and Decision Systems, Department of EECS
at the Massachusetts Institute of Technology. 
Email: \{uniesen,devavrat\}@mit.edu}
\thanks{P.~Gupta is with the Mathematics of Networks and Communications
Research Department, Bell Labs, Alcatel-Lucent.  Email:
pgupta@research.bell-labs.com}
\thanks{The work of U.~Niesen and D.~Shah was
supported in parts by DARPA grant (ITMANET) 18870740-37362-C 
and NSF grant CNS-0546590; the work of P.~Gupta
was supported in part by NSF Grants CCR-0325673 and CNS-0519535.}
}

\maketitle

\begin{abstract}
    In recent work, {\"O}zg{\"u}r, L{\'e}v{\^e}que, and Tse (2007)
    obtained a complete scaling characterization of throughput scaling
    for random extended wireless networks (i.e., $n$ nodes are placed
    uniformly at random in a square region of area $n$). They showed
    that for small path-loss exponents $\alpha\in(2,3]$ cooperative
    communication is order optimal, and for large path-loss exponents
    $\alpha > 3$ multi-hop communication is order optimal.  However,
    their results (both the communication scheme and the proof
    technique) are strongly dependent on the regularity induced with
    high probability by the random node placement.
    
    In this paper, we consider the problem of characterizing the
    throughput scaling in extended wireless networks with arbitrary node
    placement. As a main result, we propose a more general novel
    cooperative communication scheme that works for arbitrarily placed
    nodes.  For small path-loss exponents $\alpha \in (2,3]$, we show
    that our scheme is order optimal for all node placements, and
    achieves exactly the same throughput scaling as in {\"O}zg{\"u}r et
    al. This shows that the regularity of the node placement does not
    affect the scaling of the achievable rates for $\alpha\in (2,3]$.
    The situation is, however, markedly different for large path-loss
    exponents $\alpha >3$. We show that in this regime the scaling of
    the achievable per-node rates depends crucially on the regularity of
    the node placement. We then present a family of schemes that
    smoothly ``interpolate'' between multi-hop and cooperative
    communication, depending upon the level of regularity in the node
    placement. We establish order optimality of these schemes under
    adversarial node placement for $\alpha > 3$.
\end{abstract}

\begin{keywords}
    Arbitrary node placement, 
    capacity scaling,
    cooperative communication,
    hierarchical relaying, 
    multi-hop communication,
    wireless networks.
\end{keywords}

\section{Introduction}
\label{sec:intro}

Consider a wireless network with $n$ nodes placed on $[0,\sqrt{n}]^2$
(usually referred to as an \emph{extended network}), with each node being
the source for one of $n$ source-destination pairs and the destination
for another pair. The performance of this network is captured by
$\rho^*(n)$, the largest uniformly achievable rate of communication
between these source-destination pairs. While the scaling behavior of
$\rho^*(n)$ as the number of nodes $n$ goes to infinity is by now well
understood for random node placement, little is known for the case of
arbitrary node placements. In this paper, we are interested in analyzing
the impact of such arbitrary node placement on the scaling of
$\rho^*(n)$.

\subsection{Related Work}

The problem of determining the scaling of $\rho^*(n)$ was first analyzed
by Gupta and Kumar in~\cite{gup}. They show that, under random
placement of nodes in the region, certain models of communication
motivated by current technology, and random source-destination pairing,
the maximum achievable per-node rate $\rho^*(n)$ can scale at most as
$O(n^{-1/2})$.  Moreover, it was shown that multi-hop communication can
achieve essentially the same order of scaling. 

Since \cite{gup}, the problem has received a considerable amount of
attention. One stream of work~\cite{xie,jov,lev,xue,xie2,fra,ozg} has
progressively broadened the conditions on the channel model and the
communication model, under which multi-hop communication is order
optimal. Specifically, with a power loss of $r^{-\alpha}$ for signals
sent over distance $r$, it has been established that under \emph{high}
signal attenuation $\alpha>3$ and random node placement, the best
achievable per-node rate $\rho^*(n)$ for random source-destination
pairing scales essentially like $\Theta(n^{-1/2})$ and that this scaling
is achievable with multi-hop communication.

Another stream of  work \cite{gup2,xie3,kra,aer,ozg} has proposed
progressively refined multi-user cooperative schemes, which have been
shown to significantly out-perform multi-hop communication in certain
environments. In an exciting recent work, {\"O}zg{\"u}r et
al.~\cite{ozg} have shown that with nodes placed uniformly at random,
and with \emph{low} signal attenuation $\alpha\in(2,3]$, a cooperative
communication scheme can perform significantly better than multi-hop
communication. More precisely, they show that for $\alpha\in(2,3]$, the
best achievable per-node rate for random source-destination pairing
scales as $\rho^*(n)= O(n^{1-\alpha/2+\varepsilon})$ and cooperative
communication achieves a per-node rate of
$\Omega(n^{1-\alpha/2-\varepsilon})$ (here, $\varepsilon>0$ is an
arbitrary but fixed constant). That is, cooperative communication is
essentially order optimal in the attenuation regime $\alpha \in (2,3]$. 

In summary, for random extended networks with random source-destination
pairing, the optimal communication scheme exhibits the following
threshold behavior: for $\alpha \in (2,3]$ the cooperative communication
scheme is order optimal, while for $\alpha > 3$ the multi-hop
communication scheme is order optimal.

\subsection{Our Contributions}

The characterization of the scaling of $\rho^*(n)$ as a function of the
path-loss exponent $\alpha$ mentioned in the last paragraph depends
critically on the regularity induced with high probability by
placing the nodes uniformly at random. However, a wireless network
encountered in practice might not exhibit this amount of regularity. Our
interest is therefore in understanding the impact of the node placement
on the scaling of $\rho^*(n)$. To this end, we consider wireless
networks with arbitrary (i.e., deterministic) node placement (with
minimum-separation constraint).

The impact of this arbitrary node placement depends crucially on the
path-loss exponent $\alpha$. For small path-loss exponents
$\alpha\in(2,3]$, we show that for random source-destination pairing,
the rate of the best communication scheme is upper bounded as $\rho^*(n)
= O(\log^6(n)n^{1-\alpha/2})$. We then present a novel cooperative
communication scheme that achieves for any path-loss exponent $\alpha >
2$ a per-node rate of $\rho^\tsf{HR}(n) \geq n^{1-\alpha/2-o(1)}$. Thus,
our cooperative communication scheme is essentially order optimal for
any such arbitrary network with $\alpha \in (2,3]$. In other words, in
the small path-loss regime, the scaling of $\rho^*(n)$ is the same
irrespective of the regularity of the node placement.

The situation is, however, quite different for large path-loss exponents
$\alpha > 3$.  We show that in this regime the scaling of $\rho^*(n)$
depends crucially on the regularity of the node placement, and multi-hop
communication may not be order optimal for any value of $\alpha$. In
fact, for less regular networks we need more complicated cooperative
communication schemes to achieve optimal network performance. Towards
that end, we present a family of communication schemes that smoothly
``interpolate'' between cooperative communication and multi-hop
communication, and in which nodes communicate at scales that vary
smoothly from local to global. The amount of ``interpolation'' between
the cooperative and multi-hop schemes depends on the level of regularity
of the underlying node placement.  We establish the optimality of this
family of schemes for all $\alpha > 3$ under adversarial node placement.

In summary, for $\alpha\in(2,3]$ the regularity of the node placement
has no impact on the scaling of $\rho^*(n)$. Cooperative communication
is order optimal in this regime and achieves the same scaling as in the
case of random node placement. For $\alpha>3$ the regularity of the node
placement strongly impacts the scaling of $\rho^*(n)$, and a
communication scheme ``interpolating'' between multi-hop and cooperative
communication depending on the regularity of the node placement is order
optimal (under adversarial node placement). In particular, simple
multi-hop communication may not be order optimal for any $\alpha>3$.
This contrasts with the case of random node placement where multi-hop
communication is order optimal for all $\alpha>3$.

\subsection{Organization}

The remainder of this paper is organized as follows.
Section~\ref{sec:model} describes in detail the communication model.
Section~\ref{sec:results} provides formal statements of our results.
Sections~\ref{sec:desc} and \ref{sec:schemes_multihop} describe our new
cooperative communication scheme (for the $\alpha\in (2,3]$ regime) and
``interpolation'' scheme (for the $\alpha> 3$ regime) for arbitrary
wireless networks. Sections~\ref{sec:analysis} through
\ref{sec:adversaryx} contain proofs.  Finally,
Sections~\ref{sec:discussion} and~\ref{sec:conclusions} contain
discussions and concluding remarks.

\section{Model}
\label{sec:model}

In this section, we introduce some notational conventions and describe
in detail the network and channel models.

We use the following conventions: $K_i$ for different $i$ denote
strictly positive finite constants independent of $n$. Vectors and
matrices are denoted by boldface whenever the vector or matrix structure
is of importance. We denote by $(\cdot)^T$ and $(\cdot)^\dagger$
transpose and conjugate transpose, respectively.  To simplify notation,
we assume, when necessary, that fractions are integers and omit
$\ceil{\cdot}$ and $\floor{\cdot}$ operators.

Consider the square
\begin{equation*}
    A(n) \defeq [0,\sqrt{n}]^2
\end{equation*}
of area $n$, and let $V(n)\subset A(n)$ be a
set of $\abs{V(n)} = n$ nodes on\footnote{The setting considered here
with $n$ nodes placed on a square of area $n$ is called an
\emph{extended} network. If the $n$ nodes are placed on a square of unit
area, we speak of a \emph{dense} network. While dense networks are not
treated in detail in this paper, we briefly discuss implications of the
results for the dense setting in Section \ref{sec:dense}.} $A(n)$. We
say that $V(n)$ has \emph{minimum-separation} $r_{\min}$ if $r_{u,v}\geq
r_{\min}$ for all $u,v\in V(n)$, where $r_{u,v}$ is the Euclidean
distance between nodes $u$ and $v$. We use the same channel model as
in~\cite{ozg}. Namely, the (sampled) received signal at node $v$ is
\begin{equation}
    \label{eq:channel}
    y_v[t] = \sum_{u\in V(n)\setminus\{v\}}h_{u,v}[t]x_u[t]+z_v[t]
\end{equation}
for all $v\in V(n)$, and where $\{x_u[t]\}_{u,t}$ are the (sampled)
signals sent by the nodes in $V(n)$.  Here $\{z_v[t]\}_{v,t}$ are 
independent and identically distributed (i.i.d.)
with distribution $\mc{N}_{\mbb{C}}(0,1)$ (i.e., circularly symmetric
complex Gaussian with mean $0$ and variance $1$), and 
\begin{equation*}
    h_{u,v}[t] = r_{u,v}^{-\alpha/2}\exp(\sqrt{-1}\theta_{u,v}[t]),
\end{equation*} 
for \emph{path-loss exponent} $\alpha>2$. We assume that for each
$t\in\mbb{N}$, the phases $\{\theta_{u,v}[t]\}_{u,v}$ are
i.i.d.\footnote{It is worth pointing out that recent work \cite{fra2}
suggests that, under certain assumptions on scattering elements, for
$\alpha\in(2,3)$, and for very large values of $n$, the i.i.d. phase
assumption as a function of $u,v\in V(n)$ used here is too optimistic.
However, subsequent work by the same authors \cite{fra3} shows that
under different assumptions on the scatterers, the channel model used
here is still valid even for $\alpha\in(2,3)$, and for very large values
of $n$. This indicates that the question of channel modeling for very
large networks in the low path-loss regime is somewhat delicate and
requires further investigation. We point out that for $\alpha\geq 3$
this issue does not arise.} with uniform distribution on
$[0,2\pi)$. We either assume that for each $u,v\in V(n)$ the random
process $\{\theta_{u,v}[t]\}_{t}$ is stationary ergodic in $t$, which is
called \emph{fast fading} in the following, or that for each $u,v\in
V(n)$ the random process $\{\theta_{u,v}[t]\}_{t}$ is constant in $t$,
which is called \emph{slow fading} in the following.  In either case, we
assume full channel state information (CSI) is available at all nodes,
i.e., each node knows all $\{\theta_{u,v}[t]\}_{u,v}$ at time $t$. While
the full CSI assumption is quite strong, it can be shown that
availability of a $2$-bit \emph{quantized} version of
$\{\theta_{u,v}[t]\}_{u,v}$ at all nodes is sufficient for the
achievable schemes presented here (see Section \ref{sec:csi} for the
details). We also impose an average power constraint of $1$ on the
signal $\{x_u[t]\}_{t}$ for every node $u\in V(n)$.

Each node $u\in V(n)$ wants to transmit information at uniform rate
$\rho(n)$ to some other node $w\in V(n)$. We call $u$ the \emph{source}
and $w$ the \emph{destination} node of this communication pair. The set
of all communication pairs can be described by a \emph{traffic matrix}
$\lambdauc(n)\in\{0,1\}^{n\times n}$, where the entry in $\lambdauc(n)$
corresponding to $(u,w)$ is equal to $1$ if node $u$ is a source for
node $w$.  We say that $\lambdauc(n)$ is a \emph{permutation traffic
matrix} if it is a permutation matrix (i.e., every node is a source for
exactly one communication pair and a destination for exactly one
communication pair). For a traffic matrix $\lambdauc(n)$, let $\rho^*(n)$
be the highest rate of communication that is uniformly achievable for
each source-destination pair. For a permutation traffic matrix
$\lambdauc(n)$, $\rho^*(n)$ can also be understood as the maximal
achievable per-node rate.

\section{Main Results}
\label{sec:results}

This section presents the formal statement of our results. The results
are divided into two parts. In Section \ref{sec:low}, we consider low
path-loss exponents, i.e., $\alpha\in(2,3]$.  We present a cooperative
communication scheme for arbitrary node placement and for either fast or
slow fading. We show that this communication scheme is order optimal for
all node placements when $\alpha\in(2,3]$. In Section \ref{sec:high}, we
consider high path-loss exponents, i.e., $\alpha >3$.  We present a
communication scheme that ``interpolates'' between the cooperative and
the multi-hop communication schemes, depending on the regularity of the
node placement.  We show that this communication scheme is order optimal
under adversarial node placement with regularity constraint when $\alpha
> 3$.

\subsection{Low Path Loss Regime $\alpha\in(2,3]$}
\label{sec:low}

The first result proposes a novel communication scheme, called
\emph{hierarchical relaying} in the following, and bounds the per-node
rate $\rho^\tsf{HR}(n)$ that it achieves. This provides a lower bound to
$\rho^*(n)$, the largest achievable per-node rate. The hierarchical
relaying scheme enables cooperative communication on the scale of the
network size. In the random node placement case, this cooperation could
be enabled in a cluster around the source node (cooperatively
transmitting) and in a cluster around its destination node
(cooperatively receiving). With arbitrary node placement, such an
approach does no longer work, as both the source as well as the
destination nodes may be isolated. The hierarchical relaying scheme
circumvents this issue by relaying data between each source-destination
pair over a densely populated region in the network. A detailed
description of this scheme is provided in Section \ref{sec:desc}, the
proof of Theorem \ref{thm:achievability} is contained in Section
\ref{sec:achievability}.

\begin{theorem}
    \label{thm:achievability}
    Under fast fading, for any $\alpha >2$, $r_{\min}\in(0,1)$, and
    $\delta\in(0,1/2)$, there exists 
    \begin{equation*}
        b_1(n) \geq n^{-O\big(\log^{\delta-1/2}(n)\big)}
    \end{equation*}
    such that for any $n$, node
    placement $V(n)$ with minimum separation $r_{\min}$, and permutation
    traffic matrix $\lambdauc(n)$, we have
    \begin{equation*}
        \rho^*(n) 
        \geq \rho^\textup{\tsf{HR}}(n)
        \geq b_1(n)n^{1-\alpha/2}.
    \end{equation*}
    The same conclusion holds for slow fading with probability 
    at least
    \begin{equation*}
     1-\exp\Big(-2^{\Omega\big(\log^{1/2+\delta}(n)\big)}\Big)
     = 1-o(1)
    \end{equation*}
    as $n\to\infty$. 
\end{theorem}
Theorem \ref{thm:achievability} shows that the per-node rate
$\rho^{\tsf{HR}}(n)$ achievable by the hierarchical relaying scheme is
at least $n^{1-\alpha/2-\beta(n)}$, where the ``loss'' term $\beta(n)$
converges to zero as $n\to\infty$ at a rate arbitrarily close to
$O\big(\log^{-1/2}(n)\big)$ (by choosing $\delta$ small).
The performance of the hierarchical relaying scheme can intuitively be
understood as follows. As mentioned before, the scheme achieves
cooperation on a global scale. This leads to a multi-antenna gain of
order $n$. On the other hand, communication is over a distance of order
$n^{1/2}$, leading to a power loss of order $n^{-\alpha/2}$. Combining
these two factors results in a per-node rate of $n^{1-\alpha/2}$.

We note that Theorem~\ref{thm:achievability} remains valid under
somewhat weaker conditions than having minimum separation
$r_{\min}\in(0,1)$.  Specifically, we show that the result of
{\"O}zg{\"u}r et al. \cite{ozg} can be recovered through Theorem
\ref{thm:achievability} as the random node placement satisfies these
weaker conditions.  We discuss this in more detail in
Section~\ref{sec:dmin}. 

The next theorem establishes optimality of the hierarchical relaying
scheme in the range of $\alpha \in (2,3]$ for arbitrary node placement.
The proof of the theorem is presented in Section \ref{sec:converse}.

\begin{theorem}
    \label{thm:converse}
    Under either fast or slow fading, for any $\alpha\in(2,3]$,
    $r_{\min}\in(0,1)$, there exists $b_2(n) = O\big(\log^6(n)\big)$
    such that for any $n$, node placement $V(n)$ with minimum separation
    $r_{\min}$, and for $\lambdauc(n)$ chosen uniformly at random from the
    set of all permutation traffic matrices, we have
    \begin{equation*}
        \rho^*(n) 
        \leq b_2(n)n^{1-\alpha/2} 
    \end{equation*}
    with probability $1-o(1)$ as $n\to\infty$.
\end{theorem}

Note that Theorem \ref{thm:converse} holds only with probability
$1-o(1)$ for different reasons for the slow and fast fading case. For
fast fading, this is due to the randomness in the selection of the
permutation traffic matrix. In other words, for fast fading, with high
probability we select a traffic matrix for which the theorem holds. For
the slow fading case, there is additional randomness due to the fading
realization. Here, with high probability we select a traffic matrix and
we experience a fading for which the theorem hold. 

Comparing Theorems~\ref{thm:achievability} and~\ref{thm:converse}, we
see that for $\alpha\in(2,3]$ the proposed hierarchical relaying scheme
is order optimal, in the sense that
\begin{equation*}
    \lim_{n\to\infty}\frac{\log(\rho^{\tsf{HR}}(n))}{\log(n)}
    = \lim_{n\to\infty}\frac{\log(\rho^{*}(n))}{\log(n)}
    = 1-\alpha/2.
\end{equation*}
Moreover, the rate it achieves is the same order as is
achievable in the case of randomly placed nodes.
Hence in the low path-loss regime $\alpha\in(2,3]$, the heterogeneity
caused by the arbitrary node placement has no effect on achievable
communication rates.

\subsection{High Path Loss Regime $\alpha > 3$}
\label{sec:high}

We now turn to the high path-loss regime $\alpha>3$. In the case of
\emph{randomly} placed nodes, multi-hop communication achieves a
per-node rate of $\rho^{\tsf{MH}}(n) = \Omega(n^{-1/2})$ with
probability $1-o(1)$ and is order optimal for $\alpha> 3$. For
\emph{arbitrarily} placed nodes, the situation is quite different as
Theorem~\ref{thm:adversary} shows. The proof of Theorem
\ref{thm:adversary} is contained in Section \ref{sec:adversary}.

\begin{theorem}
    \label{thm:adversary}
    Under either fast or slow fading, for any $\alpha >3$,
    for any $n$, there exists a node placement $V(n)$ with minimum
    separation $1/2$ such that for $\lambdauc(n)$ chosen uniformly at
    random from the set of all permutation traffic matrices, we have
    \begin{align*}
        \rho^*(n) & \leq 2^{2+5\alpha}n^{1-\alpha/2}, \\
        \rho^\textup{\tsf{MH}}(n) & \leq 4^\alpha n^{-\alpha/2},
    \end{align*}
    as $n\to\infty$ with probability $1-o(1)$.
\end{theorem}

Comparing Theorem~\ref{thm:adversary} with Theorem
\ref{thm:achievability} shows that under adversarial node placement with
minimum-separation constraint the hierarchical relaying scheme is order
optimal even when $\alpha>3$. Moreover, Theorem~\ref{thm:adversary}
shows that there exist node placements satisfying a minimum separation
constraint for which hierarchical relaying achieves a rate of at least a
factor of order $n$ higher than multi-hop communication for any
$\alpha>3$. In other words, for those node placements cooperative
communication is necessary for order optimality also for any $\alpha>
3$, in stark contrast to the situation with random node placement, where
multi-hop communication is order optimal for all $\alpha >3$. 

Theorem~\ref{thm:adversary} suggests that it is the level of regularity
of the node placement that decides what scheme to choose for path-loss
exponent $\alpha > 3$. So far, we have seen two extreme cases: For
random node placement, resulting in very regular node placements
with high probability, only local cooperation is necessary and multi-hop
is an order-optimal communication scheme. For adversarial arbitrary node
placement, resulting in a very irregular node placement, global
cooperation is necessary and hierarchical relaying is an order-optimal
communication scheme. We now make this notion of regularity precise, and
show that, depending on the regularity of the node placement, an
appropriate ``interpolation'' between multi-hop and hierarchical
relaying is required for $\alpha > 3$ to achieve the optimal
performance. We refer to this ``interpolation'' scheme as
\emph{cooperative multi-hop} communication in the following.

Before we state the result, we need to introduce some notation.
Consider again a node placement $V(n)\subset A(n)$ with minimum
separation $r_{\min}\in(0,1)$. Divide $A(n)$ into squares of sidelength
$d(n)\leq\sqrt{n}$, and fix a constant $\mu\in(0,1]$. We say that
\emph{$V(n)$ is $\mu$-regular at resolution $d(n)$} if every such square
contains at least $\mu d^2(n)$ nodes. Note that every node placement is
trivially $1$-regular at resolution $\sqrt{n}$; a random node placement
can be shown to be $\mu$-regular at resolution $\log(n)$ with
probability $1-o(1)$ as $n\to\infty$ for any $\mu <1$; and nodes that
are placed on each point in the integer lattice inside $A(n)$ are 
$1$-regular at resolution $1$. 

The cooperative multi-hop scheme enables cooperative communication on
the scale of regularity $d(n)$. Neighboring squares of sidelength $d(n)$
cooperatively communicate with each other. To transmit between a source
and its destination, we use multi-hop communication over those squares.
In other words, we use cooperative communication at small scale $d(n)$, and
multi-hop communication at large scale $\sqrt{n}$. For regular node placements,
i.e., $d(n)=1$, the cooperative multi-hop scheme becomes the classical
multi-hop scheme. For very irregular node placement, i.e.,
$d(n)=n^{1/2}$, the cooperative multi-hop scheme becomes the
hierarchical relaying scheme discussed in the last section. 

The next theorem provides a lower bound on the per-node rate
$\rho^{\tsf{CMH}}(n)$ achievable with the cooperative multi-hop scheme.
The proof of the theorem can be found in Section \ref{sec:multihop}.

\begin{theorem}
    \label{thm:achievabilityx}
    Under fast fading, for any $\alpha >2$, $r_{\min}\in(0,1)$,
    $\mu\in(0,1)$, and $\delta\in(0,1/2)$ there exists 
    \begin{equation*}
        b_3(n) \geq n^{-O\big(\log^{\delta-1/2}(n)\big)}
    \end{equation*}
    such that for any $n$, node placement $V(n)$ with minimum separation
    $r_{\min}$, and permutation traffic matrix $\lambdauc(n)$, we have
    \begin{equation*}
        \rho^*(n) 
        \geq \rho^\textup{\tsf{CMH}}(n)
        \geq b_3(n){d^*}^{3-\alpha}(n)n^{-1/2},
    \end{equation*}
    where
    \begin{equation*}
        d^*(n)
        \defeq \min\{h: \text{$V(n)$ \emph{is $\mu$ regular at resolution} $h$}\}.
    \end{equation*}
    The same conclusion holds for slow fading with probability $1-o(1)$
    as $n\to\infty$. 
\end{theorem}

Theorem~\ref{thm:achievabilityx} shows that if $V(n)$ is regular at
resolution $d^*(n)$ then a per-node rate of at least $\rho^\tsf{CMH}(n)
\geq {d^*}^{3-\alpha}(n)n^{-1/2-\beta(n)}$ is achievable, where, as
before, the ``loss'' term $\beta(n)$ converges to zero as $n\to\infty$
at a rate arbitrarily close to $O\big(\log^{-1/2}(n)\big)$. The
performance of the cooperative multi-hop scheme can intuitively be
understood as follows. The scheme achieves cooperation on a scale of
$d^2(n)$. This leads to a multi-antenna gain of order $d^2(n)$. On the
other hand, communication is over a distance of order $d(n)$, leading to
a power loss of order $d^{-\alpha}(n)$. Moreover, each
source-destination pair at a distance of order $n^{1/2}$ must transmit
their data over order $n^{1/2}d^{-1}(n)$ many hops, leading to a
multi-hop loss of $n^{-1/2}d(n)$. Combining these three factors results
in a per-node rate of $d^{3-\alpha}(n)n^{-1/2}$.

The next theorem shows that Theorem~\ref{thm:achievabilityx} is tight
under adversarial node placement under a constraint on the regularity.
The proof of the theorem is presented in Section \ref{sec:adversaryx}.

\begin{theorem}
    \label{thm:adversaryx}
    Under either fast or slow fading, for any $\alpha >3$, there exists
    $b_4(n) = O\big(\log^6(n)\big)$, such that for any $n$, and $d^*(n)$,
    there exists a node placement $V(n)$ with minimum separation $1/2$
    and $1/2$-regular at resolution $d^*(n)$ such that for $\lambdauc(n)$ chosen
    uniformly at random from the set of all permutation traffic
    matrices, we have
    \begin{equation*}
        \rho^*(n) \leq b_4(n){d^*}^{3-\alpha}(n)n^{-1/2},
    \end{equation*}
    with probability $1-o(1)$ as $n\to\infty$.
\end{theorem}

As an example, assume that
\begin{equation*}
    d^*(n) = n^{\eta}
\end{equation*}
for some $\eta\geq 0$. Then Theorem~\ref{thm:achievabilityx} shows that
for any node placement of regularity $d^*(n)$ and $\alpha>3$,
\begin{equation*}
    \rho^\tsf{CMH}(n) \geq n^{(3-\alpha)\eta-1/2-\beta(n)},
\end{equation*}
where $\beta(n)$ converges to zero as $n\to\infty$ at a rate arbitrarily
close to $O\big(\log^{-1/2}(n)\big)$. In other words
\begin{equation*}
    \lim_{n\to\infty}\frac{\log(\rho^\tsf{CMH}(n))}{\log(n)} \geq (3-\alpha)\eta-1/2.
\end{equation*}
Moreover, by Theorem \ref{thm:adversaryx} there exist node
placements with same regularity such that for random permutation traffic
with high probability $\rho^*(n)$ is (essentially) of the same order,
in the sense that
\begin{equation*}
    \lim_{n\to\infty}\frac{\log(\rho^*(n))}{\log(n)} \leq (3-\alpha)\eta-1/2.
\end{equation*}
In particular, for $\eta = 0$ (i.e., regular node placement), and for
$\eta = \log\log(n)/\log(n)$ (i.e., random node placement), we
obtain the order $n^{-1/2}$ scaling as expected. For $\eta = 1/2$ (i.e.,
completely irregular node placement), we obtain the order
$n^{1-\alpha/2}$ scaling as in Theorems \ref{thm:achievability} and
\ref{thm:adversary}.

\section{Hierarchical Relaying Scheme}
\label{sec:desc}

This section describes the architecture of our hierarchical relaying
scheme. On a high level, the construction of this scheme is as follows.
Consider $n$ nodes $V(n)$ placed arbitrarily on the square region $A(n)$
with a minimum separation $r_{\min}$. Divide $A(n)$ into squarelets
of equal size. Call a squarelet \emph{dense}, if it contains a number of
nodes proportional to its area. For each source-destination pair, choose
such a dense squarelet as a \emph{relay}, over which it will transmit
information (see Figure \ref{fig:relay}).

\begin{figure}[!ht]
    \begin{center}
        \input{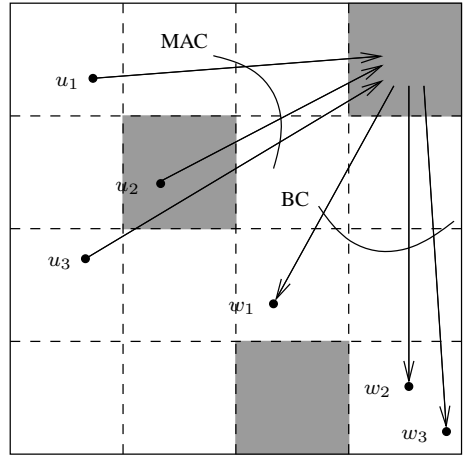}
    \end{center}

    \caption{Sketch of one level  of the hierarchical relaying scheme.
    Here $\{(u_i,w_i)\}_{i=1}^3$ are three source-destination pairs.
    Groups of source-destination pairs relay their traffic over dense
    squarelets, which contain a number of nodes proportional to their
    area (shaded). We time share between the different dense squarelets
    used as relays. Within all these relay squarelets the scheme is used
    recursively to enable joint decoding and encoding at each relay.
    }

    \label{fig:relay}
\end{figure}

Consider now one such relay squarelet and the nodes that are
transmitting information over it. If we assume for the moment that all
the nodes within the same relay squarelet could cooperate then we would
have a multiple access channel (MAC) between the source nodes and the
relay squarelet, where each of the source nodes has one transmit
antenna, and the relay squarelet (acting as one node) has many receive
antennas.  Between the relay squarelet and the destination nodes, we
would have a broadcast channel (BC), where each destination node has one
receive antenna, and the relay squarelet (acting again as one node) has
many transmit antennas. The cooperation gain from using this kind of
scheme arises from the use of multiple antennas for these multiple
access and broadcast channels.

To actually enable this kind of cooperation at the relay squarelet,
local communication within the relay squarelets is necessary. It can be
shown that this local communication problem is actually the same as the
original problem, but at a smaller scale. Hence we can use the same
scheme recursively to solve this subproblem. We terminate the recursion
after several iterations, at which point we use simple TDMA to bootstrap
the scheme.

The construction of the hierarchical relaying scheme is presented in
detail in Section \ref{sec:construction}. A back-of-the-envelope
calculation of the per-node rate it achieves is presented in Section
\ref{sec:rates}. A detailed analysis of the hierarchical relaying scheme
is presented in Sections \ref{sec:analysis} and \ref{sec:achievability}.

\subsection{Construction}
\label{sec:construction}

Recall that
\begin{equation*}
    A(b)\defeq [0,\sqrt{b}]^2
\end{equation*}
is the square region of area $b$. The scheme described here assumes that
$n$ nodes are placed arbitrarily in $A(n)$ with minimum separation
$r_{\min}\in(0,1)$.  We want to find some rate, say $\rho_0$, that can
be supported for all $n$ source-destination pairs of a given permutation
traffic matrix $\lambdauc(n)$. The scheme that is described below is
``recursive'' (and hence hierarchical) in the following sense.  In order
to achieve rate $\rho_0$ for $n$ nodes in $A(n)$, it will use as a
building block a scheme for supporting rate $\rho_1$ for a network of 
\begin{equation*}
    n_1 \defeq \frac{n}{2 \gamma(n)}
\end{equation*}
nodes over $A(a_1)$ (square of area $a_1$) with 
\begin{equation*}
    a_1 \defeq \frac{n}{\gamma(n)}
\end{equation*}
for any permutation traffic matrix $\lambdauc(n_1)$ of $n_1$ nodes.  Here the
\emph{branching factor} $\gamma(n)$ is a function such that $\gamma(n)
\to \infty$ as $n\to \infty$. We will optimize over the choice of
$\gamma(n)$ later.  The same construction is used for the scheme over
$A(a_1)$, and so on. In general, our scheme does the following at level
$\ell \geq 0$  of the hierarchy (or recursion). In order to achieve rate
$\rho_\ell$ for any permutation traffic matrix $\lambdauc(n_\ell)$ over 
\begin{equation*}
    n_\ell \defeq \frac{n}{2^\ell \gamma^\ell(n)}
\end{equation*}
nodes in $A(a_\ell)$, with
\begin{equation*}
    a_\ell \defeq \frac{n}{\gamma^\ell(n)},
\end{equation*}
use a scheme achieving rate $\rho_{\ell+1}$ over $n_{\ell +1}$ nodes in
$A(a_{\ell +1})$  for any permutation traffic matrix $\lambdauc(n_{\ell+1})$.
The recursion is terminated at some level $L(n)$ to be chosen later. 

We now describe how the hierarchy is constructed between levels $\ell$
and $\ell +1$ for $0\leq \ell < L(n)$. Each source-destination pair
chooses some squarelet as a relay over which it transmits its message.
This relaying of messages takes place in two phases
-- a multiple access phase and a broadcast phase. We first describe the
selection of relay squarelets, then the operation of the network during
the multiple access and broadcast phases, and finally the termination of
the hierarchical construction.

\subsubsection{Setting up Relays} 

Given $n_\ell$ nodes in $A(a_\ell)$, divide the square region
$A(a_\ell)$ into $\gamma(n)$ equal sized squarelets.  Denote them by
$\{A_k(a_{\ell+1})\}_{k=1}^{\gamma(n)}$.  Call a squarelet \emph{dense}
if it contains at least $n_\ell/2\gamma(n) = n_{\ell+1}$ nodes. In other
words, a dense squarelet contains a number of nodes of at least a
$1/2^{\ell+1}$ fraction of its area.  We show that since the nodes in
$A(a_\ell)$ have constant minimum separation $r_{\min}$, a squarelet can
contain at most $O(a_{\ell+1})$ (i.e.  $O(a_\ell/\gamma(n))$) nodes, and
hence that there are at least $\Theta(2^{-\ell}\gamma(n))$ dense
squarelets. Each source-destination pair chooses a dense squarelet such
that both the source and the destination are at a distance
$\Omega(\sqrt{a_{\ell+1}})$ from it. We call this dense squarelet the
\emph{relay} of this source-destination pair. We show that the relays
can be chosen such that each relay squarelet has at most $n_{\ell+1}$
communication pairs that use it as relay, and we assume this worst case
in the following discussion.

\subsubsection{Multiple Access Phase} 

Source nodes that are assigned to the same (dense) relay squarelet send
their messages simultaneously to that relay. We time share between the
$\Theta(2^{-\ell}\gamma(n))$ different relay squarelets.  If the nodes
in the relay squarelet could cooperate, we would be dealing with a MAC
with at most $n_{\ell+1}$ transmitters, each with one antenna, and one
receiver with at least $n_{\ell+1}$ antennas. In order to achieve this
cooperation, we use a hierarchical (i.e., recursive) construction. For
this recursive construction, assume that we have access to a
communication scheme to transmit data according to a permutation traffic
matrix $\lambdauc(n_{\ell+1})$ between $n_{\ell+1}$ nodes located in a
square of area $a_{\ell+1}$. We now show how this scheme at scale
$a_{\ell+1}$ can be used to construct a scheme for scale $a_{\ell}$ (see
Figure~\ref{fig:mac}).

\begin{figure*}[!ht]
    \hspace{-0.25cm}
    \scalebox{0.97}{
    \input{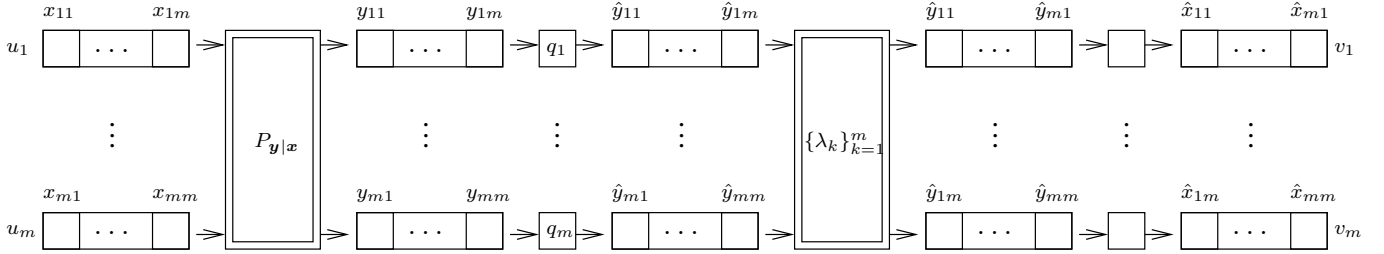}
    }
    \caption{Description of the multiple access phase at level $\ell$ in
    the hierarchy with  $m\defeq n_{\ell+1}$.  The first system block
    represents the wireless channel, connecting source nodes
    $\{u_i\}_{i=1}^{n_{\ell+1}}$ with relay nodes
    $\{v_i\}_{i=1}^{n_{\ell+1}}$. The second system block are quantizers
    $\{q_i\}_{i=1}^{n_{\ell+1}}$ used at the relay nodes. The third
    system block represents using $n_{\ell+1}$ times the communication
    scheme at level $\ell+1$ (organized as $n_{\ell+1}$ permutation
    traffic matrices $\{\lambdauc_k(n_{\ell+1})\}_{k=1}^{n_{\ell+1}}$) to
    ``transpose'' the matrix of quantized observations
    $\{\hat{y}_{ij}\}_{i,j=1}^{n_{\ell+1}}$. In other words, before the
    third system block, node $v_1$ has access to
    $\{\hat{y}_{1j}\}_{j=1}^{n_{\ell+1}}$, and after the third system
    block, node $v_1$ has access to
    $\{\hat{y}_{i1}\}_{i=1}^{n_{\ell+1}}$.  The fourth system block are
    matched filters used at the relay nodes. }

    \label{fig:mac}
\end{figure*}

Suppose there are $n_{\ell+1}$ source nodes $u_1,\dots, u_{n_{\ell+1}}$
(located anywhere in $A(a_\ell)$) that relay their message over the
$n_{\ell+1}$ relay nodes $v_1,\dots, v_{n_{\ell+1}}$ (located in the
same dense squarelet of area $a_{\ell+1}$). Each source node $u_i$
divides its message bits into $n_{\ell+1}$ parts of equal length.
Denote by $x_{ij}$ the encoded part $j$ of the message bits of node
$u_i$ ($x_{ij}$ is really a large sequence of channel symbols; to
simplify the exposition, we shall, however, assume it is only a single
symbol). The message parts corresponding to
$\{x_{ij}\}_{i=1}^{n_{\ell+1}}$ will be relayed over node $v_j$, as will
become clear in the following. Sources $\{u_i\}_{i=1}^{n_{\ell+1}}$,
transmit $\{x_{ij}\}_{i=1}^{n_{\ell+1}}$ at time $j$ for $j\in
\{1,\ldots n_{\ell+1}\}$. 

Let $y_{kj}$ be the observed channel output at relay $v_k$ at time $j$.
Note that $y_{kj}$ depends only on channel inputs
$\{x_{ij}\}_{i=1}^{n_{\ell+1}}$.  In order to decode the message parts
corresponding to $\{x_{ij}\}_{i=1}^{n_{\ell+1}}$ at relay node $v_j$, it
needs to obtain the observations $\{y_{ij}\}_{i=1}^{n_{\ell+1}}$ from
all other relay nodes. In other words, all relays need to exchange
information. For this, each relay $v_k$ quantizes its observation
$\{y_{kj}\}_{j=1}^{n_{\ell+1}}$ at an appropriate rate $K$ independent
of $n$ to obtain $\{\hy_{kj}\}_{j=1}^{n_{\ell+1}}$. Quantized
observation $\hy_{kj}$ is to be sent from relay $v_k$ to relay $v_j$.
Thus, each of the $n_{\ell+1}$ relay nodes now has a message of size $K$
for every other relay node. 

This communication demand within the relay squarelet can be organized as
$n_{\ell+1}$ permutation traffic matrices
$\{\lambdauc_j(n_{\ell+1})\}_{j=1}^{n_{\ell+1}}$ between the $n_{\ell+1}$
relay nodes. Note that these relay nodes are located in the same square
of area $a_{\ell+1}$.  In other words, we are now faced with the
original problem, but at smaller scale $a_{\ell+1}$. Therefore, using
$n_{\ell+1}$ times the assumed scheme for transmitting according to a
permutation traffic matrix for $n_{\ell+1}$ nodes in $A(a_{\ell+1})$,
relay $v_j$ can obtain all quantized observations
$\{\hy_{ij}\}_{i=1}^{n_{\ell+1}}$.  Now $v_j$ uses $n_{\ell+1}$ matched
filters on $\{\hy_{ij}\}_{i=1}^{n_{\ell+1}}$ to obtain estimates
$\{\hat{x}_{ij}\}_{i=1}^{n_{\ell+1}}$ of
$\{x_{ij}\}_{i=1}^{n_{\ell+1}}$. In other words, each node $v_j$
computes\footnote{Note that, since we assume full CSI, node $v_j$ has
access to the channel gains $\{h_{u_i,v_k}[j]\}_{i,k}$ at any time $t\geq j$.
In particular, this is the case at the time the matched filtering is
performed.}
\begin{equation*}
    \hat{x}_{ij} = \sum_{k=1}^{n_{\ell+1}} 
    \frac{h_{u_i,v_k}^\dagger[j]}
    {\sqrt{\sum_k \lvert h_{u_i,v_k}[j]\rvert^2}}\hat{y}_{kj}
\end{equation*}
for every $i\in\{1,\ldots,n_{\ell+1}\}$.  Using these estimates it then
decodes the messages corresponding to $\{x_{ij}\}_{i=1}^{n_{\ell+1}}$.

\subsubsection{Broadcast Phase} 

Nodes in the same relay squarelet then send their decoded messages
simultaneously to the destination nodes corresponding to this relay.  We
time share between the different relay squarelets. If the nodes in the
relay squarelet could cooperate, we would be dealing with a BC with one
transmitter with at least $n_{\ell+1}$ antennas and with at most
$n_{\ell+1}$ receivers, each with one antenna. In order to achieve this
cooperation, a similar hierarchical construction as for the MAC phase is
used. As in the MAC phase, assume that we have access to a scheme to
transmit data according to a permutation traffic matrix
$\lambdauc(n_{\ell+1})$ between $n_{\ell+1}$ nodes located in a square
of area $a_{\ell+1}$. We again use this scheme at scale
$a_{\ell+1}$ in the construction of the scheme for scale $a_{\ell}$ (see
Figure~\ref{fig:bc}).

\begin{figure*}[!ht]
    \hspace{-0.25cm}
    \scalebox{0.97}{
    \input{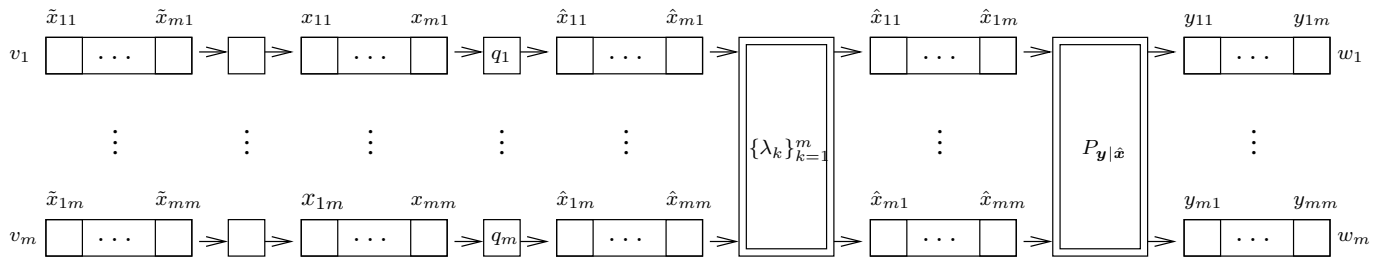}
    }

    \caption{Description of the broadcast phase at level $\ell$ in the
    hierarchy with $m\defeq n_{\ell+1}$.  The first system block
    represents transmit beamforming at each of the relay nodes
    $\{v_i\}_{i=1}^{n_{\ell+1}}$. The second system block are quantizers
    $\{q_i\}_{i=1}^{n_{\ell+1}}$ used at the relay nodes.  The third
    system block represents using $n_{\ell+1}$ times the communication
    scheme at level $\ell+1$ (organized as $n_{\ell+1}$ permutation
    traffic matrices $\{\lambdauc_k(n_{\ell+1})\}_{k=1}^{n_{\ell+1}}$) to
    ``transpose'' the matrix of quantized beamformed channel symbols
    $\{\hat{x}_{ij}\}_{i,j=1}^{n_{\ell+1}}$. In other words, before the
    third system block, node $v_1$ has access to
    $\{\hat{x}_{i1}\}_{i=1}^{n_{\ell+1}}$, and after the third system
    block, node $v_1$ has access to
    $\{\hat{x}_{1j}\}_{j=1}^{n_{\ell+1}}$. The fourth system block is
    the wireless channel, connecting relay nodes
    $\{v_i\}_{i=1}^{n_{\ell+1}}$ with destination nodes
    $\{w_i\}_{i=1}^{n_{\ell+1}}$.}

    \label{fig:bc}
\end{figure*}

Suppose there are $n_{\ell+1}$ relay nodes $v_1,\dots, v_{n_{\ell+1}}$
(located in the same dense squarelet of area $a_{\ell+1}$) that relay
traffic for $n_{\ell+1}$ destination nodes $w_1,\dots, w_{n_{\ell+1}}$
(located anywhere in $A(a_{\ell})$). Recall that at the end of the MAC
phase, each relay node $v_j$ has (assuming decoding was successful)
access to parts $j$ of the message bits of all source nodes
$\{u_i\}_{i=1}^{n_{\ell+1}}$. Node $v_j$ re-encodes these parts
independently; call $\{\tilde{x}_{ij}\}_{i=1}^{n_{\ell+1}}$ the encoded
channel symbols (as before, we assume $\tilde{x}_{ij}$ is only a single
symbol to simplify exposition). Relay node $v_j$ then performs transmit
beamforming on $\{\tilde{x}_{ij}\}_{i=1}^{n_{\ell+1}}$ for the
$n_{\ell+1}$ transmit antennas of $\{v_k\}_{k=1}^{n_{\ell+1}}$ to be
sent at time $T+j$ (for some appropriately chosen $T>0$ not depending on
$j$). Call $x_{kj}$ the resulting channel symbol to be sent from relay
node $v_k$. Then\footnote{Note that, since we only assume causal CSI,
relay node $v_j$ does not actually have
access to $\{h_{v_k,w_i}[T+j]\}_{k,i}$ at the time the beamforming is
performed. This problem can, however, be circumvented. The details are
provided in the proofs (see Lemma \ref{thm:bc}).}
\begin{equation*}
    x_{kj} = \sum_i \frac{h_{v_k,w_i}^\dagger[T+j]}
    {\sqrt{\sum_k \lvert h_{v_k,w_i}[T+j]\rvert^2}}\tilde{x}_{ij}.
\end{equation*}
In order to actually send this channel symbol, relay node $v_k$ needs to
obtain $x_{kj}$ from node $v_j$. Thus, again all relay nodes need to
exchange information. 

To enable local cooperation within the relay squarelet, each relay node
$v_j$ quantizes its beamformed channel symbols
$\{x_{kj}\}_{k=1}^{n_{\ell+1}}$ at an appropriate rate $K\log(n)$ with
$K$ independent of $n$ to obtain $\{\hat{x}_{kj}\}_{k=1}^{n_{\ell+1}}$.
Now, quantized value $\hat{x}_{kj}$ is sent from relay $v_j$ to relay
$v_k$.  Thus, each of the $n_{\ell+1}$ relay nodes now has a message of
size $K\log(n)$ for every other relay node. 

This communication demand within the relay squarelet can be organized as
$n_{\ell+1}$ permutation traffic matrices
$\{\lambdauc_k(n_{\ell+1})\}_{k=1}^{n_{\ell+1}}$ between the $n_{\ell+1}$
relay nodes. Note that these relay nodes are located in the same square
of area $a_{\ell+1}$.  Hence, we are again faced with the original
problem, but at smaller scale $a_{\ell+1}$. Using $n_{\ell+1}$ times the
assumed scheme for transmitting according to a permutation traffic
matrix for $n_{\ell+1}$ nodes in $A(a_{\ell+1})$, relay $v_k$ can obtain
all quantized beamformed channel symbols
$\{\hat{x}_{kj}\}_{j=1}^{n_{\ell+1}}$.  Now each $v_k$ sends
$\hat{x}_{kj}$ over the wireless channel at time instance $T+j$ (with
$T$ chosen to account for the preceding MAC phase and the local
cooperation in the BC phase). Call $y_{ij}$ the received channel output
at destination node $w_i$ at time instance $T+j$.  Using $y_{ij}$,
destination node $w_i$ can now decode part $j$ of the message bits of
its source node $u_i$.

\subsubsection{Spatial Re-Use and Termination of Recursion} 

The scheme does appropriately weighted time-division among different
levels $0\leq \ell \leq L(n)$.  Within any level $\ell \geq 1$, multiple
regions of the original square $A(n)$ of area $n$ are being operated in
parallel.  The details related to the effects of interference between
different regions operating at the same level of hierarchy are discussed
in the proofs. 

The recursive construction terminates at some large enough level $L =
L(n)$ (to be chosen later).  At this scale, we have $n_L$ nodes in area
$A(a_L)$. A permutation traffic matrix at this level comprises $n_L$
source-destination pairs. These transmissions are performed using TDMA.
Again, multiple regions in the original square of area $n$ at level $L$
are active simultaneously.

\subsection{Achievable Rates}
\label{sec:rates}

Here we present a back-of-the-envelope calculation of the per-node
rate $\rho^{\tsf{HR}}(n)$ achievable with the hierarchical relaying scheme
described in the previous section. The complete proof is stated in Section
\ref{sec:achievability}. We assume throughout that long block codes and
corresponding optimal decoders are used for transmission. 

Instead of computing the rate achieved by hierarchical relaying, it will
be convenient to instead analyze its inverse, i.e., the time utilized
for transmission of a single message bit from each source to its
destination under a permutation traffic matrix $\lambdauc(n)$. Using the
hierarchical relaying scheme, each message travels through $L$ levels
of the hierarchy. Call $\tau_\ell(n)$ the amount of time spent for the
transmission of one message bit between each of the $n_\ell$
source-destination pairs  at level $\ell$ in the hierarchy. We compute
$\tau_\ell(n)$ recursively. 

At any level $\ell \geq 1$, there are multiple regions of area $a_\ell$
operating at the same time. Due to the spatial re-use, each of these
regions gets to transmit a constant fraction of time. It can be shown
that the addition of interference due to this spatial re-use leads only
to a constant loss in achievable rate.  Hence the time required to send
one message bit is only a constant factor higher than the one needed if
region $A(a_\ell)$ is considered separately. Consider now one such
region $A(a_\ell)$. By the time sharing construction, only one of its
$\Theta(2^{-\ell}\gamma(n))$ dense relay squarelets of area $a_{\ell+1}$
is active at any given moment. Hence the time required to operate all
relay squarelets is a $\Theta(2^{-\ell}\gamma(n))$ factor higher than
for just one relay squarelet separately. Consider now one such relay
squarelet, and assume $n_{\ell+1}$ source nodes in $A(a_{\ell})$
communicate each $n_{\ell+1}$ message bits to their respective
destination nodes through a MAC phase and BC phase with the help of the
$n_{\ell+1}$ relay nodes in this relay squarelet of area $a_{\ell+1}$. 

In the MAC phase, each of the $n_{\ell+1}$ sources simultaneously sends
one bit to each of the $n_{\ell+1}$ relay nodes. The total time for
this transmission is composed of two terms.
\begin{enumerate}
    \item[i)] Transmission of $n_{\ell+1}$ message bits from each of the
        $n_{\ell+1}$ source nodes to those many relay nodes. Since we
        time share between $\Theta(2^{-\ell}\gamma(n))$ relay
        squarelets, we can transmit with an average power constraint of
        $\Theta(2^{-\ell}\gamma(n))$ during the time a relay squarelet
        is active, and still satisfies the overall average power
        constraint of $1$. With this ``bursty'' transmission strategy,
        we require a total of
        \begin{equation}
            \label{eq:mac1}
            O\bigg(n_{\ell+1} \frac{a_\ell^{\alpha/2}}{2^{-\ell}\gamma(n)n_{\ell+1}}\bigg) \\
            = O\big(n_{\ell+1}4^\ell \gamma^{\ell(1-\alpha/2)}(n)n^{\alpha/2-1}\big)
        \end{equation}
        channel uses to transmit $n_{\ell+1}$ bits per source node. The
        terms on the left-hand side of~\eqref{eq:mac1} can be understood
        as follows: $n_{\ell+1}$ is the number of bits to be
        transmitted; $a_\ell^{\alpha/2}$ is the power loss since most
        nodes communicate over a distance of $\Theta(a_\ell^{1/2})$;
        $2^{-\ell}\gamma(n)$ is the average transmit power;
        $n_{\ell+1}$ is the multiple-antenna gain, since we have
        that many transmit and receive antennas.
    \item[ii)] We show that constant rate quantization of the received
        observations at the relays is sufficient. Hence the $n_{\ell+1}$
        bits for all sources generate $O(n_{\ell+1})$ transmissions at
        level $\ell+1$ of the hierarchy. Therefore, 
        \begin{equation}
            \label{eq:mac2}
            O(n_{\ell+1} \tau_{\ell+1}(n))
        \end{equation}
        channel uses are needed to communicate all quantized
        observations to their respective relay nodes.
\end{enumerate}
Combining~\eqref{eq:mac1} and~\eqref{eq:mac2}, accounting for the factor
$2^{-\ell}\gamma(n)$ loss due to time division between relay
squarelets, we obtain that the transmission time for one message bit
from each source to the relay squarelet in the MAC phase at level $\ell$
is
\begin{equation}
    \label{eq:be0}
    \tau_\ell^{\tsf{MAC}}(n) 
    = O\Big(2^\ell \gamma^{1+\ell(1-\alpha/2)}(n) n^{\alpha/2-1} + \tau_{\ell+1}(n)\Big).
\end{equation}

Next, we compute the number of channel uses per message bit received by the
destination nodes in the BC phase.  Similar to the MAC phase, each of the
$n_{\ell+1}$ relay nodes has $n_{\ell+1}$ message bits out of which one
bit is to be transmitted to each of the $n_{\ell+1}$ destination nodes.
Since there are $n_{\ell+1}$ relay nodes, each destination node receives
$n_{\ell+1}$ message bits. As before the required transmission time has
two components.
\begin{enumerate}
    \item[i)] Transmission of the encoded and quantized message bits from each
        of the $n_{\ell+1}$ relay nodes to all other relay nodes at
        level $\ell+1$ of the hierarchy. We show that each message bit
        results in $O\big( (\ell+1)\log n\big)$ quantized bits.
        Therefore, $O\big(n_{\ell+1} (\ell+1)\log n\big)$ bits need to
        be transmitted from each relay node. This requires 
        \begin{equation}
            \label{eq:bc1}
            O\big(n_{\ell+1} (\ell+1) \log (n) \tau_{\ell+1}(n)\big)
        \end{equation}
        channel uses.
    \item[ii)] Transmission of $n_{\ell+1}$ message bits from the relay nodes to 
        each destination node. As before, we use bursty transmission with an
        average power constraint of $\Theta(2^{-\ell}\gamma(n))$
        during the fraction $\Theta(2^{\ell}\gamma^{-1}(n))$ of time
        each relay squarelet is active (this satisfies the overall
        average power constraint of $1$). Using this bursty strategy
        requires
        \begin{equation}
            \label{eq:bc2}
            O\bigg(n_{\ell+1} \frac{a_\ell^{\alpha/2}}{2^{-\ell}\gamma(n)n_{\ell+1}}\bigg) \\
            = O\big(n_{\ell+1}4^\ell \gamma^{\ell(1-\alpha/2)}(n)n^{\alpha/2-1}\big)
        \end{equation}
        channel uses for transmission of $n_{\ell+1}$ bits per
        destination node. As in the MAC phase, $n_{\ell+1}$ in the left
        hand side of~\eqref{eq:bc2} can be understood as the number of
        bits to be transmitted, $a_\ell^{\alpha/2}$ as the power loss
        for communicating over distance $\Theta(a_\ell^{1/2})$,
        $2^{-\ell}\gamma(n)$ as the average transmit power, and
        $n_{\ell+1}$ as the multiple-antenna gain.
\end{enumerate}
Combining~\eqref{eq:bc1} and~\eqref{eq:bc2}, accounting for a factor 
$2^{-\ell}\gamma(n)$ loss due to time division between relay squarelets,
the transmission time for one message bit from the relays
to each destination node in the BC phase at level $\ell$ is
\begin{equation}
    \label{eq:be1}
    \tau_\ell^{\tsf{BC}}(n) 
    = O\Big(2^\ell \gamma^{1+\ell(1-\alpha/2)}(n) n^{\alpha/2-1} \\
    + (\ell+1) \log (n) \tau_{\ell+1}(n)\Big).
\end{equation}

From \eqref{eq:be0} and \eqref{eq:be1}, we obtain the following recursion
\begin{align}
    \label{eq:be2}
    \tau_\ell(n) 
    & = \tau_\ell^{\tsf{MAC}}(n) + \tau_\ell^{\tsf{BC}}(n) \nonumber\\
    & = O\Big(2^\ell \gamma^{\ell(1-\alpha/2) + 1}(n) n^{\alpha/2-1} 
    + (\ell+1) \log (n) \tau_{\ell+1}(n)\Big) \nonumber\\
    & = O\Big(2^L \gamma(n) n^{\alpha/2-1} 
    + L\log (n) \tau_{\ell+1}(n)\Big),
\end{align}
where we have used $\alpha>2$. This recursion holds for all $0 \leq \ell
< L$. At level $L$, we use TDMA among $n_L$ nodes in region $A(a_L)$
with a permutation traffic matrix $\lambdauc(n_L)$. Each of the $n_L$
source-destination pairs uses the wireless channel for $1/n_L$ fraction
of the time at power $O(n_L)$, satisfying the average power constraint.
Assuming the received power is less than $1$ for all $n$ (so that we
operate in the power limited regime), we can achieve a rate of at least
$\Omega(a_L^{-\alpha/2})$ between any source-destination pair.
Equivalently
\begin{align}
    \label{eq:be3}
    \tau_L(n) 
    & = O(a_L^{\alpha/2}) \nonumber\\
    & = O\big(n^{\alpha/2} \gamma^{-L\alpha/2}(n)\big) \nonumber\\
    & = O\big(n^{\alpha/2} \gamma^{-L}(n)\big).
\end{align}

Combining \eqref{eq:be2} and \eqref{eq:be3}, we have
\begin{align}
    \label{eq:be4}
    \tau_0(n) 
    & = O\big(n^{\alpha/2-1}2^L\gamma(n)+ L \log(n) \tau_1(n)\big) \nonumber\\
    & = \ldots \nonumber\\
    & = O\Big(n^{\alpha/2-1}\big(L\log(n)\big)^L 2^L\gamma(n) 
    + \big(L\log(n)\big)^L \tau_L(n)\Big) \nonumber\\
    & = O\Big(n^{\alpha/2-1}\big(L\log(n)\big)^L \big(2^L\gamma(n) 
    + n \gamma^{-L}(n)\big)\Big).
\end{align}
The term
\begin{equation*}
    \big(L\log(n)\big)^L \big(2^L\gamma(n) + n \gamma^{-L}(n)\big)
\end{equation*}
is the ``loss'' factor over the desired order $n^{\alpha/2-1}$ scaling,
and we now choose the branching factor $\gamma(n)$ and the hierarchy 
depth $L\defeq L(n)$ to make it small. Fix a $\delta\in(0,1/2)$ and set
\begin{align*}
    L(n) & \defeq \log^{1/2-\delta}(n), \\
    \gamma(n) & \defeq n^{1/(L(n)+1)}. 
\end{align*}
With this 
\begin{align*}
    \big(L(n)\log(n)\big)^{L(n)} 
    & \leq n^{2\log^{-1/2-\delta}(n)\log\log(n)}, \\
    2^{L(n)}\gamma(n)
    & \leq n^{\log^{-1/2-\delta}(n)+\log^{\delta-1/2}(n)}, \\
    n\gamma^{-L(n)}(n)
    & \leq n^{\log^{\delta-1/2}(n)}.
\end{align*}
Since $\delta>0$, the $n^{\log^{\delta-1/2}(n)}$ term dominates in 
\eqref{eq:be4}, and we obtain 
\begin{align*}
    \tau_0(n) 
    & \leq \tilde{b}(n)n^{\alpha/2 - 1}, 
\end{align*}
where 
\begin{equation*}
    \tilde{b}(n) \leq n^{O(\log^{\delta-1/2}(n))}.
\end{equation*}
Hence the per-node rate of the hierarchical relaying scheme is lower
bounded as
\begin{equation*}
    \rho^{\tsf{HR}}(n) 
    = 1/\tau_0(n) 
    \geq b(n)n^{1-\alpha/2},
\end{equation*}
with 
\begin{equation*}
    b(n) \geq n^{-O(\log^{\delta-1/2}(n))}.
\end{equation*}
Note that to minimize the loss term, we should choose
$\delta>0$ to be small.

\section{Cooperative Multi-Hop Scheme}
\label{sec:schemes_multihop}

In this section, we provide a brief description of the cooperative
multi-hop scheme. The details of the construction and the analysis of its
performance can be found in Section \ref{sec:multihop}.

Recall that a node placement $V(n)$ is $\mu$-regular at resolution
$d(n)$ if every square $[id(n),(i+1)d(n)]\times[jd(n),(j+1)d(n)]$ for
some $i,j\in\mbb{N}$ contains at least $\mu d^2(n)$ nodes. Given such a
node placement $V(n)$, divide it into squares of sidelength $d(n)$.
Consider four adjacent squares, combined into a bigger square of
sidelength $2d(n)$. By the regularity assumption on $V(n)$, this bigger
square contains at least $4\mu d^2(n)$ nodes. Hence we can apply the
hierarchical relaying scheme introduced in the last section to support
any permutation traffic within this bigger square at a per-node rate of
\begin{equation*}
    b(n)(d^2(n))^{1-\alpha/2} = b(n)d^{2-\alpha}(n),
\end{equation*}
where $b(n)$ is essentially of order $n^{-\log^{-1/2}(n)}$. By properly
choosing the permutation traffic matrices within every possible such
bigger square of sidelength $2d(n)$, this creates a equivalent
communication graph with $n/d^2(n)$ nodes each corresponding to a square
of sidelength $d(n)$ in $A(n)$, and with edges between nodes
corresponding to neighboring squares. With the above communication
procedure and appropriate spatial re-use, each such edge has a capacity
of 
\begin{equation*}
    d^2(n) b(n)d^{2-\alpha}(n)=b(n)d^{4-\alpha}(n). 
\end{equation*}
The resulting communication graph is depicted in Figure \ref{fig:backbone2}.
\begin{figure}[!ht]
    \begin{center}
        \input{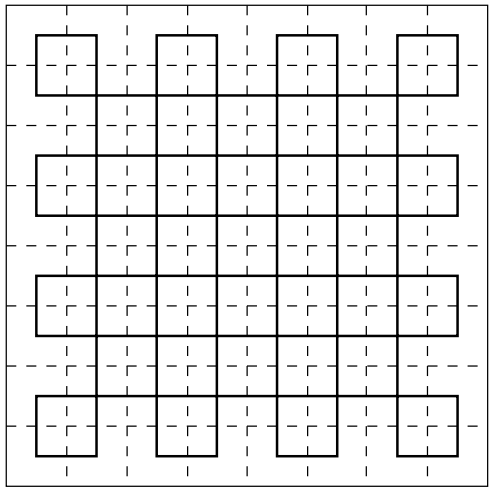}
    \end{center}

    \caption{Communication graph (in bold) resulting from the
    construction of the cooperative multi-hop scheme. The entire square
    has sidelength $\sqrt{n}$, and the dashed squares have sidelength
    $d(n)$. Each (bold) edge in the communication graph corresponds to
    using the hierarchical relaying scheme between the nodes in the
    adjacent squares of sidelength $d(n)$.}

    \label{fig:backbone2}
\end{figure}

Now, to send a message from a source node in $V(n)$ to its destination
node, we first locate the squares of sidelength $d(n)$ they are located
in. We then route the message over the edges of the communication graph
constructed above in a multi-hop fashion. By the construction of the
communication graph, each such edge is implemented using the
hierarchical relaying scheme. In other words, we perform multi-hop
communication over distance $\sqrt{n}$ with hop length $d(n)$, and each
such hop is implemented using hierarchical relaying over distance
$d(n)$.  Since each edge in the communication graph has a capacity of
$b(n)d^{4-\alpha}(n)$ and has to support roughly $n^{1/2}d(n)$
source-destination pairs, we obtain a per-node rate of 
\begin{align*}
    \rho^{\tsf{CMH}}(n) 
    & \geq b(n)d^{4-\alpha}(n)n^{-1/2}d^{-1}(n)  \\
    & = b(n)d^{3-\alpha}(n)n^{-1/2}
\end{align*}
per source-destination pair.

\section{Analysis of the Hierarchical Relaying Scheme}
\label{sec:analysis}

In this section, we analyze in detail the hierarchical relaying scheme.
Throughout Sections~\ref{sec:relay} to \ref{sec:bc}, we consider
communication at level $\ell$, $0\leq \ell< L=L(n)$, of the hierarchy.
All constants $K_i$ are independent of $\ell$.  

Recall that at level $\ell$, we have a square region $A(a_{\ell})$ of
area 
\begin{equation*}
    a_{\ell}\defeq \frac{n}{\gamma^\ell(n)}
\end{equation*}
containing 
\begin{equation*}
    n_{\ell} \defeq \frac{n}{2^\ell \gamma^\ell(n)}
\end{equation*}
nodes $V(n_\ell)$. We divide $A(a_\ell)$ into $\gamma(n)$ squarelets of
area $a_{\ell+1}$.  Recall that a squarelet of area $a_{\ell+1}$ in
level $\ell$ of the hierarchy is called dense if it contains at least
$n_{\ell+1}$ nodes.  We impose a power constraint of $P_\ell(n) =
\Theta(2^{-\ell}\gamma(n))$ during the time any particular relay
squarelet is active. Since we time share between
$\Theta(2^{-\ell}\gamma(n))$ relay squarelets, this satisfies the
overall average power constraint (by choosing constants appropriately). 

Since other regions of area $a_\ell$ are active at the same time as the
one under consideration, we have to deal with interference.  To this
end, we consider a slightly more general noise model that includes the
experienced interference at the relay squarelets. More precisely, we
assume that, for all $u\in V(n_\ell)$, the additive noise term
$\{z_u[t]\}_{t}$ is independent of the signal $\{x_u[t]\}_{t}$ and of
the channel gains $\{h_{u,v}[t]\}_{v,t}$; that the noise term is
stationary and ergodic across time $t$, but with arbitrary dependence
across nodes $u$; and that the noise has zero mean and bounded power
$N_0$ independent of $n$. Note that we do not require the additive noise
term to be Gaussian. In the above, $N_0$ accounts for both noise (which
has power $1$ in the original model), as well as interference. We show
in Section~\ref{sec:achievability} that these assumptions are valid.

Recall the following choice of $\gamma(n)$ and $L(n)$:
\begin{equation}
    \label{eq:gammadef}
    \begin{aligned}
        L(n) & \defeq \log^{1/2-\delta}(n), \\
        \gamma(n) & \defeq n^{1/(L(n)+1)},
    \end{aligned}
\end{equation}
with $\delta\in(0,1/2)$ independent of $n$. This choice satisfies
\begin{equation}
    \label{eq:assumptions}
    \begin{aligned}
        \gamma(n) & \leq \gamma(\tilde{n}) & \text{if $n\leq\tilde{n}$}, \\
        \gamma^{L(n)}(n) & \leq n & \text{for all $n$}, \\
        2^{-L(n)}\gamma(n) & \to \infty & \text{as $n\to\infty$},
    \end{aligned}
\end{equation}
The first condition in \eqref{eq:assumptions} implies that the number of
squarelets $\gamma(n)$ we divide $A(n)$ into increases in $n$. The
second condition implies the squarelet area $a_{L(n)}$ at the last level
of the hierarchy is bigger than $1$. As we shall see, the third
condition implies that the number of dense squarelets at the last level
(and hence at every level) grows unbounded as $n\to\infty$ (see Lemma
\ref{thm:structure} below). 

Throughout Section~\ref{sec:analysis}, we consider the fast fading
channel model. Slow fading is discussed in Section~\ref{sec:slow}.

\subsection{Setting up Relays}
\label{sec:relay}

The first lemma states that the minimum-separation requirement $r_{\min}\in(0,1)$
implies that a constant fraction of squarelets must be dense. We point
out that this is the only consequence of the minimum-separation
requirement used to prove Theorem~\ref{thm:achievability}.
Thus Theorem~\ref{thm:achievability} remains valid if we just assume
that Lemma~\ref{thm:structure} below holds directly. See also
Section~\ref{sec:dmin} for further details.

\begin{lemma}
    \label{thm:structure}
    For any $V(n_\ell)\subset A(a_\ell)$ with $\card{V(n_\ell)}\geq
    n_\ell$ and with minimum separation $r_{\min}\in(0,1)$,  each of its
    squarelets of area $a_{\ell+1}$ contains at most $K_1
    a_\ell/\gamma(n)$ nodes, and there are at least $K_2
    2^{-\ell}\gamma(n)$ dense squarelets.
\end{lemma}
\begin{proof}
    Put a circle of radius $r_{\min}/2$ around each node. By the
    minimum-separation requirement, these circles do not intersect. Each
    node covers an area of $\pi r_{\min}^2/4$. Increasing the sidelength of
    each squarelet by $r_{\min}$, this provides a total area of
    \begin{equation*}
        \big(\sqrt{a_\ell/\gamma(n)}+r_{\min}\big)^2
        \leq \frac{a_\ell}{\gamma(n)}(1+r_{\min})^2
    \end{equation*}
    in which the circles around these nodes are packed. Here we have used that
    $\gamma^{\ell+1}(n)\leq n$ by~\eqref{eq:assumptions}, and therefore
    \begin{equation*}
        \gamma(n) \leq n/\gamma^\ell(n) = a_\ell.
    \end{equation*}
    Hence there can be at most $K_1 a_\ell/\gamma(n)$ nodes per squarelet with 
    \begin{equation*}
        K_1 \defeq 4\frac{(1+r_{\min})^2}{\pi r_{\min}^2}.
    \end{equation*} 
    Note that, since $r_{\min} < 1$, we have $K_1 > 1$. 

    Let $d(n_\ell)$ be the number of dense squarelets in $A(a_\ell)$,
    and therefore $\gamma(n)-d(n_\ell)$ is the number of squarelets that
    are not dense. By the argument in the last paragraph, each dense
    squarelet contains at most $K_1 a_\ell/\gamma(n)$ nodes, and those
    squarelets that are not dense contain less than $n_{\ell+1}$ nodes
    by the definition of dense squarelets. Hence $d(n_\ell)$ must
    satisfy
    \begin{equation*}
        d(n_\ell) K_1 a_\ell/\gamma(n) +\big(\gamma(n)-d(n_\ell)\big)n_{\ell+1} 
        \geq \card{V(n_\ell)}
        \geq n_\ell.
    \end{equation*}
    Thus, using $a_\ell = 2^\ell n_\ell$,
    $n_{\ell+1}=n_\ell/2\gamma(n)$, we have
    \begin{equation*}
        d(n_{\ell})K_12^\ell+(\gamma(n)-d(n_{\ell}))/2 
        \geq \gamma(n).
    \end{equation*}
    As $K_12^\ell > 1$, this yields
    \begin{equation*}
        d(n_\ell)
        \geq \frac{1-1/2}{K_12^\ell-1/2}\gamma(n)
        \geq \frac{2^{-\ell}}{2K_1}\gamma(n)
        = K_2 2^{-\ell}\gamma(n),
    \end{equation*}
    with
    \begin{equation*}
        K_2 \defeq \frac{1}{2K_1}.
    \end{equation*}
\end{proof}

Consider $V(n_\ell)\subset A(a_\ell)$ with $\card{V(n_\ell)}$, and
choose arbitrary $K_2 2^{-\ell} \gamma(n)$ dense squarelets of area
$a_{\ell+1}$ (as guaranteed by Lemma~\ref{thm:structure}). Call those
squarelets $\{A_k(a_{\ell+1})\}_{k=1}^{K_2 2^{-\ell} \gamma(n)}$. For
each sour-destination pair, we now select one such dense squarelet to
relay traffic over. To avoid bottlenecks, this selection has to be done
such that all relay squarelets carry approximately the same amount of
traffic. Moreover, for technical reasons, the distances from the source
and the destination to the relay squarelet cannot be too small.

Formally, the selection of relay squarelets can be described by the
\emph{schedules} $S\in \{0,1\}^{n_\ell\times K_2 2^{-\ell}\gamma(n)}$
with $s_{u,k}=1$ if source node $u$ relays traffic over dense squarelet
$k$, and $\widetilde{S}\in\{0,1\}^{K_2 2^{-\ell}\gamma(n)\times n_\ell}$
with $\tilde{s}_{k,w}=1$ if destination node $w$ receives traffic
from dense squarelet $k$. With slight abuse of notation, let
$r_{u,A_k(a_{\ell+1})}$ be the distance between node $u\in V(n_\ell)$
and the closest point in $A_k(a_{\ell+1})$, i.e., 
\begin{equation}
    \label{eq:ruadef}
    r_{u,A_k(a_{\ell+1})} 
    \defeq \min_{v\in A_k(a_{\ell+1})} r_{u,v}.
\end{equation}
Define the sets
\begin{align}
    \label{eq:schedules}
    \mc{S}(n_\ell) & \defeq \Big\{S\in \{0,1\}^{n_\ell\times K_2 2^{-\ell}\gamma(n)}:  \nonumber\\
    & \qquad\quad 0\leq {\textstyle\sum_{u=1}^{n_\ell} s_{u,k}} \leq n_{\ell+1} \ \forall k, \nonumber\\
    & \qquad\quad 0\leq {\textstyle\sum_{k=1}^{K_2 2^{-\ell}\gamma(n)}} s_{u,k} \leq 1 \ \forall u, \nonumber\\
    & \qquad\quad s_{u,k} = 1 \text{ implies } r_{u,A_k(a_{\ell+1})}
    \geq \sqrt{2a_{\ell+1}} \ \forall u,k
    \Big\}
\end{align}
and 
\begin{equation*}
    \widetilde{\mc{S}}(n_\ell) 
    \defeq \big\{\widetilde{S}\in\{0,1\}^{K_2 2^{-\ell}\gamma(n)\times n_\ell}:
    \widetilde{S}^T\in\mc{S}(n_{\ell})\big\}.
\end{equation*}
The sets $\mc{S}(n_\ell)$ and $\widetilde{\mc{S}}(n_\ell)$ are the
collection of schedules satisfying the conditions mentioned in the last
paragraph. More precisely, the first condition in \eqref{eq:schedules}
ensures that at most $n_{\ell+1}$ source-destination pairs relay over
the same dense squarelet, the second condition ensures that each
source-destination pair chooses at most one relay squarelet, and the
third condition ensures that sources and destinations are at least at
distance $\sqrt{2a_{\ell+1}}$ from the chosen relay squarelet.

Next, we prove that any node placement that satisfies Lemma
\ref{thm:structure} allows for a decomposition of any permutation
traffic matrix $\lambdauc(n_\ell)$ into a small number of schedules belonging to
$\mc{S}(n_\ell)$ and $\widetilde{\mc{S}}(n_\ell)$. 
\begin{lemma}
    \label{thm:relay}
    There exist $K_3$ such that for all $n$ large enough (independent of
    $\ell$), and every permutation traffic matrix
    $\lambdauc(n_\ell)\in\{0,1\}^{n_\ell\times n_\ell}$ we can find $K_3 2^\ell$
    schedules $\{S^{(i)}(n_\ell)\}_{i=1}^{K_3
    2^\ell}\subset\mc{S}(n_\ell)$,
    $\{\widetilde{S}^{(i)}(n_\ell)\}_{i=1}^{K_3
    2^\ell}\subset\widetilde{\mc{S}}(n_\ell)$ satisfying
    \begin{equation*}
        \lambdauc(n_\ell) 
        = \sum_{i=1}^{K_3 2^\ell} S^{(i)}(n_\ell)\widetilde{S}^{(i)}(n_\ell).
    \end{equation*}
\end{lemma}
\begin{proof}
    Pick an arbitrary source-destination pair in $\lambdauc(n_\ell)$, and
    consider the squarelets containing the source and the destination
    node. Since each squarelet has side length $\sqrt{a_{\ell+1}}$,
    there are at most $50$ squarelets at distance less than
    $\sqrt{2a_{\ell+1}}$ from either of those two squarelets. As
    $2^{-L(n)}\gamma(n)\to\infty$ as $n\to\infty$ by
    \eqref{eq:assumptions}, there exists $K$ (independent of $\ell$)
    such that for $n\geq K$ we have $50 \leq  K_2 2^{-\ell-1}
    \gamma(n)$. Since there are at least $K_2 2^{-\ell}\gamma(n)$ dense
    squarelets by Lemma~\ref{thm:structure}, there must exist at least
    $K_2 2^{-\ell-1}\gamma(n)$ dense squarelets that are at distance at
    least $\sqrt{2a_{\ell+1}}$ from both the squarelets containing the
    source and the destination node.

    In order to construct a decomposition of $\lambdauc(n_\ell)$, we use the
    following procedure. Sequentially, each of the $n_\ell$
    source-destination pairs chooses one of the (at least) $K_2
    2^{-\ell-1}\gamma(n)$ dense squarelets at distance at least $\sqrt{2
    a_{\ell+1}}$ that has not already been chosen by $n_{\ell+1}$ other
    pairs. If any source-destination pair can not select such a
    squarelet, then stop the procedure and use the source-destination
    pairs matched with dense squarelets so far to define matrices
    $S^{(1)}(n_\ell)$ and $\widetilde{S}^{(1)}(n_\ell)$. Now, remove all
    the matched source-destination pairs, forget that dense squarelets
    were matched to any source-destination pair and redo the above
    procedure, going through the remaining source-destination pairs.  
    
    Let
    \begin{equation*}
        K_3 \defeq 4/K_2.
    \end{equation*}
    We claim that by repeating this process of generating matrices
    $S^{(i)}(n_\ell)$ and $\widetilde{S}^{(i)}(n_\ell)$, we can match
    all source-destination pairs to some dense squarelet with at most
    $K_3 2^\ell$ such matrices. Indeed, a new pair of matrices
    is generated only when a source-destination pair can not be matched
    to any of its available (at least) $K_2 2^{-\ell-1}\gamma(n)$ dense
    squarelets. If this happens, all these dense squarelets are matched
    by $n_{\ell+1}=n_\ell/2\gamma(n)$ pairs. Hence at least $K_2
    2^{-\ell-2} n_\ell$ source-destination pairs are matched in each
    ``round''. Since there are $n_\ell$ total pairs, we need at most 
    \begin{equation*}
        \frac{n_\ell}{K_2 2^{-\ell-2}n_\ell} = K_3 2^{\ell}
    \end{equation*}
    matrices $S^{(i)}(n_\ell)$ and $\widetilde{S}^{(i)}(n_\ell)$.
\end{proof}

For a permutation traffic matrix $\lambdauc(n_\ell)$, communication proceeds as follows.
Write
\begin{equation*}
    \lambdauc(n_\ell) = \sum_{i=1}^{K_3 2^\ell} S^{(i)}(n_\ell)\widetilde{S}^{(i)}(n_\ell)
\end{equation*}
as in Lemma~\ref{thm:relay}. Split time into $K_3 2^\ell$ equal length time
slots. In slot $i$, we use $S^{(i)}(n_\ell)\widetilde{S}^{(i)}(n_\ell)$ as our
traffic matrix. Consider without loss of generality $i=1$ in the
following. Write
\begin{equation*}
    S^{(1)}(n_\ell)\widetilde{S}^{(1)}(n_\ell)
    = \sum_{k=1}^{K_2 2^{-\ell}\gamma(n)}
    S^{(1,k)}(n_{\ell+1})\widetilde{S}^{(1,k)}(n_{\ell+1}),
\end{equation*}
where $S^{(1,k)}(n_{\ell+1})\widetilde{S}^{(1,k)}(n_{\ell+1})$ is the
traffic relayed over the dense squarelet $A_k(a_{\ell+1})$. We time
share between the schedules for $k\in\{1,\ldots,K_2
2^{-\ell}\gamma(n)\}$. Consider now any such $k$. In the worst case,
there are exactly $n_{\ell+1}$ communication pairs to be relayed over
$A_k(a_{\ell+1})$, and the relay squarelet $A_k(a_{\ell+1})$ contains
exactly $n_{\ell+1}$ nodes. We shall assume this worst case in the
following.

We focus on the transmission according to the traffic matrix
$S^{(1,1)}(n_{\ell+1})\widetilde{S}^{(1,1)}(n_{\ell+1})$. Let
$V(n_{\ell+1})$ be the nodes in $A_1(a_{\ell+1})$, and let
$U(n_{\ell+1})$ and $W(n_{\ell+1})$ be the source and destination nodes
of $S^{(1,1)}(n_{\ell+1})\widetilde{S}^{(1,1)}(n_{\ell+1})$,
respectively. In other words, the source nodes $U(n_{\ell+1})$
communicate to their respective destination nodes $W(n_{\ell+1})$ using
the nodes $V(n_{\ell+1})$ as relays.

\subsection{Multiple Access Phase}

Each source node in $U(n_{\ell+1})$ splits its message into $n_{\ell+1}$
equal length parts. Part $j$ at every node $u\in U(n_{\ell+1})$ is to be
relayed over the $j$-th node in $V(n_{\ell+1})$. Each part is separately
encoded at the source and separately decoded at the destination. After
the source nodes are done transmitting their messages, the nodes in the
relay squarelet quantize their (sampled) observations corresponding to
part $j$ and communicate the quantized values to the $j$-th node in the
relay squarelet. This node then decodes the $j$-th message parts of all
source nodes. Note that this induces a uniform traffic pattern between
the nodes in the relay squarelet, i.e., every node needs to transmit
quantized observations to every other node. While this traffic pattern
does not correspond to a permutation traffic matrix, it can be written
as a sum of $n_{\ell+1}$ permutation traffic matrices. A fraction
$1/n_{\ell+1}$ of the traffic within the relay squarelet is transmitted
according to each of these permutation traffic matrices.  This setup is
depicted in Figure~\ref{fig:mac} in Section~\ref{sec:construction}.

Assuming for the moment that we have a scheme to send the quantized
observations to the dedicated node in the relay squarelet, the traffic
matrix $S^{(1,1)}(n_{\ell+1})$ between $U(n_{\ell+1})$ and
$V(n_{\ell+1})$ describes then a MAC with $n_{\ell+1}$ transmitters,
each with one antenna, and one receiver with $n_{\ell+1}$ antennas. We
call this the \emph{MAC induced by $S^{(1,1)}(n_{\ell+1})$} in the
following. Before we analyze the rate achievable over this induced MAC,
we need an auxiliary result on quantized channels.

\begin{figure}[!ht]
    \begin{center}
        \input{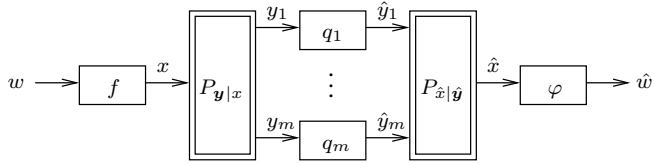}
    \end{center}
    \caption{Sketch of the quantized channel. $f$ and $\varphi$ are
    the channel encoder and decoder, respectively; $\{q_k\}_{k=1}^m$ are
    quantizers; $P_{\bm{y}|x}$ and $P_{\hat{x}|\hat{\bm{y}}}$ represent
    stationary ergodic channels with the indicated marginal distributions.}
    \label{fig:quant}
\end{figure}

Consider the quantized channel in Figure~\ref{fig:quant}. Here, $f$ is
the channel encoder, $\varphi$ the channel decoder, $\{q_k\}_{k=1}^m$
quantizers. All these have to be chosen.  $P_{\bm{y}|x}$ and
$P_{\hat{x}|\hat{\bm{y}}}$, on the other hand, represent fixed
stationary ergodic channels with the indicated marginal distributions.
We call $R$ the rate of the channel code $(f,\varphi)$ and
$\{R_k\}_{k=1}^m$ the rates of quantizers $\{q_k\}_{k=1}^m$.

\begin{lemma}
    \label{thm:quant}
    If there exist distributions $P_x$ and
    $\{P_{\hat{y}_k|y_k}\}_{k=1}^m$ such that $R<I(x;\hat{x})$ and $R_k
    > I(y_k;\hat{y}_k),\, \forall k$, then $\big(R,\{R_k\}_{k=1}^m\big)$ is
    achievable over the quantized channel.
\end{lemma}
\begin{proof}
    The proof follows from a simple extension of Theorem 1 in Appendix
    II of~\cite{ozg}.
\end{proof}

\begin{lemma}
    \label{thm:mac}
    Let the additive noise $\{z_v\}_{v\in V(n_{\ell+1})}$ be
    uncorrelated (over $v$).  For the MAC induced by
    $S^{(1,1)}(n_{\ell+1})$ with per-node average power constraint
    $P_\ell(n)\leq n_{\ell+1}^{-1}a_\ell^{\alpha/2}$, a 
    rate of
    \begin{equation*}
        \rho^\textup{\tsf{MAC}}_\ell(n) 
        \geq K_4 P_\ell(n)n_{\ell+1}a_\ell^{-\alpha/2}
    \end{equation*}
    per source node is achievable, and the number of bits required at
    each relay node to quantize the observations is at most $K_5$ bits
    per $n_{\ell+1}$ total message bits\footnote{Total message bits
    refers to the sum of all message bits transmitted by the
    $n_{\ell+1}$ source nodes.} sent by the source nodes.
\end{lemma}
\begin{proof}
    The source nodes send signals with a power of (essentially)
    $n_{\ell+1}^{-1}a_\ell^{\alpha/2}$ for a fraction $P_\ell(n)
    n_{\ell+1}a_\ell^{-\alpha/2}\leq 1$ of time and are silent for the
    remaining time. To ensure that interference is uniform, the time
    slots during which the nodes send signals are chosen randomly as
    follows. Generate independently for each region $A(a_\ell)$ a
    Bernoulli process $\{B[t]\}_{t\in\mbb{N}}$ with parameter
    $P_\ell(n) n_{\ell+1}a_\ell^{-\alpha/2}/(1+\eta)\leq 1$ for some
    small $\eta>0$. The nodes in $A(a_\ell)$ are active 
    whenever $B[t]=1$ and remain silent otherwise. Since the
    blocklength of the codes used is assumed to be large, this satisfies
    the average power constraint of $P_\ell(n)$ with high probability
    for any $\eta>0$. Since we are interested only in the scaling of
    capacity, we ignore the additional $1/(1+\eta)$ term in the
    following to simplify notation. Clearly, we only need to consider
    the fraction of time during which $B[t] = 1$. 
    
    Let $\bm{y}$ be the
    received vector at the relay squarelet, $\hat{\bm{y}}$ the
    (componentwise) quantized observations. We use a matched filter at
    the relay squarelet, i.e., 
    \begin{equation*}
        \hat{x}_u = \frac{\bm{h}_u^\dagger}{\norm{\bm{h}_u}}\hat{\bm{y}},
    \end{equation*}
    where column vector $\bm{h}_u = \{h_{u,v}\}_{v\in V(n_{\ell+1})}$
    are the channel gains between node $u\in U(n_{\ell+1})$ and the
    nodes in the relay squarelet $V(n_{\ell+1})$. The use of a matched
    filter is possible since we assume full CSI is available at all the nodes.

    We now use Lemma~\ref{thm:quant} to show that we can design
    quantizers $\{q_v\}_{v\in V(n_{\ell+1})}$ of constant rate and
    achieve a per-node communication rate of at least $K_4
    P_\ell(n)n_{\ell+1}a_\ell^{-\alpha/2}$. The first channel in
    Lemma~\ref{thm:quant} (see Figure~\ref{fig:quant}) will correspond
    to the wireless channel between a source node $u$ and its relay
    squarelet $V(n_{\ell+1})$. The second ``channel'' in
    Lemma~\ref{thm:quant} will correspond to the matched filter used at
    the relay squarelet. To apply Lemma~\ref{thm:quant}, we need to find
    a distribution for $x_u$ and for $\hat{y}_v|y_v$. Define
    \begin{equation*}
        \tilde{r}_u
        \defeq r_{u,A_1(a_{\ell+1})}/\sqrt{2a_\ell}
        \leq 1
    \end{equation*}
    with $r_{u,A_1(a_{\ell+1})}$ as in \eqref{eq:ruadef}, to be the
    normalized distance of the source node $u\in U(n_{\ell+1})$ to the
    relay squarelet $A_1(a_{\ell+1})$. For each $u\in U(n_{\ell+1})$ let
    $x_u\sim\mc{N}_{\mbb{C}}(0,\tilde{r}_u^{\alpha}n_{\ell+1}^{-1}a_\ell^{\alpha/2})$
    independent of $x_{\tilde{u}}$ for $u\neq \tilde{u}$, and let
    $\hat{y}_v = y_v+\tilde{z}_v$ for
    $\tilde{z}_v\sim\mc{N}_{\mbb{C}}(0,\Delta^2)$ independent of
    $\bm{y}$ and for some $\Delta^2>0$. Note that the channel input
    $x_u$ has power that depends on the normalized distance
    $\tilde{r}_u$ (i.e., only nodes $u\in U(n_{\ell+1})$ that are at
    maximal distance $\sqrt{2a_{\ell}}$ from the relay squarelet
    transmit at full available power). This is to ensure that all
    signals are received at roughly the same strength by the relays.

    We proceed by computing the mutual informations $I(y_v;
    \hat{y}_v\vert\{h_{\tilde{u},\tilde{v}}\})$ and
    $I(x_u;\hat{x}_u\vert\{h_{\tilde{u},\tilde{v}}\})$ as required in
    Lemma \ref{thm:quant} (the conditioning on
    $\{h_{\tilde{u},\tilde{v}}\}$ being due to the availability of full
    CSI). Note first that by construction of $S^{(1,1)}(n_{\ell+1})$
    (see \eqref{eq:schedules}), we have for $u\in U(n_{\ell+1})$ and
    $v\in V(n_{\ell+1})$
    \begin{equation*}
        r_{u,A_1(a_{\ell+1})}
        \leq r_{u,v}
        \leq 2r_{u,A_1(a_{\ell+1})},
    \end{equation*}
    and hence
    \begin{equation}
        \label{eq:rtildebounds}
        \frac{1}{2\sqrt{2a_{\ell}}} 
        \leq \frac{\tilde{r}_u}{r_{u,v}}
        \leq \frac{1}{\sqrt{2a_{\ell}}}.
    \end{equation}
    From this, and since $\abs{h_{u,v}}^2=r_{u,v}^{-\alpha}$, we obtain
    \begin{equation}
        \label{eq:hbounds}
        \begin{aligned}
            2^{-3\alpha/2}a_\ell^{-\alpha/2}
            & \leq \abs{h_{u,v}}^2 \tilde{r}_u^{\alpha}
            \leq 2^{-\alpha/2}a_\ell^{-\alpha/2}, \\
            2^{-3\alpha/2}n_{\ell+1}a_\ell^{-\alpha/2}
            & \leq \norm{\bm{h}_{u}}^2 \tilde{r}_u^{\alpha}
            \leq 2^{-\alpha/2}n_{\ell+1}a_\ell^{-\alpha/2}.
        \end{aligned}
    \end{equation}

    We start by computing $I(y_v; \hat{y}_v\vert\{h_{\tilde{u},\tilde{v}}\})$. We have
    \begin{equation*}
        \hat{y}_v 
        = \sum_{u\in U(n_{\ell+1})}h_{u,v}x_u+z_v+z_{\tilde{v}},
    \end{equation*}
    and hence $\hat{y}_v$ has mean zero and variance
    \begin{align*}
        \E (\abs{\hat{y}_v}^2)
        & = \sum_{u\in U(n_{\ell+1})}\abs{h_{u,v}}^2
        \tilde{r}_u^{\alpha}n_{\ell+1}^{-1}a_\ell^{\alpha/2}  +N_0+\Delta^2 \\
        & \leq n_{\ell+1} 2^{-\alpha/2} a_{\ell}^{-\alpha/2}
        n_{\ell+1}^{-1}a_\ell^{\alpha/2} +N_0+\Delta^2 \\
        & =  2^{-\alpha/2}+N_0+\Delta^2,
    \end{align*}
    where we have used \eqref{eq:hbounds}. Hence
    \begin{align}
        \label{eq:i1}
        I(y_v; \hat{y}_v\vert\{h_{\tilde{u},\tilde{v}}\})
        & = h(\hat{y}_v\vert\{h_{\tilde{u},\tilde{v}}\})
        -h(\hat{y}_v|y_v, \{h_{\tilde{u},\tilde{v}}\}) \nonumber\\
        & \leq \log\big(2\pi e \E(\abs{\hat{y}_v}^2) \big)
        -\log(2\pi e\Delta^2) \nonumber\\
        & \leq \log\big(2\pi e(2^{-\alpha/2}+N_0+\Delta^2)\big)
        -\log(2\pi e\Delta^2) \nonumber\\
        & = \log\Big(1+\frac{2^{-\alpha/2}+N_0}{\Delta^2}\Big).
    \end{align}

    We now compute $I(x_u;\hat{x}_u\vert\{h_{\tilde{u},\tilde{v}}\})$. We have
    \begin{equation*}
        \hat{x}_u 
        = \norm{\bm{h}_u}x_u+\sum_{\tilde{u}\in U(n_{\ell+1})\setminus\{u\}}
        \frac{\bm{h}_u^\dagger\bm{h}_{\tilde{u}}}{\norm{\bm{h}_u}}x_{\tilde{u}}
        +\frac{\bm{h}_u^\dagger}{\norm{\bm{h}_u}}(\bm{z}+\tilde{\bm{z}}).
    \end{equation*}
    Conditioned on $\{\bm{h}_{\tilde{u}}\}_{\tilde{u}\in U(n_{\ell+1})}$,
    \begin{equation*}
        \norm{\bm{h}_u}x_u
        \sim \mc{N}_{\mbb{C}}\big(0,\norm{\bm{h}_u}^2\tilde{r}_u^\alpha
        n_{\ell+1}^{-1}a_\ell^{\alpha/2} \big),
    \end{equation*}
    and
    \begin{equation*}
        \E\Big( \Big|
        {\textstyle\sum_{\tilde{u}\in U(n_{\ell+1})\setminus\{u\}}}
        \frac{\bm{h}_u^\dagger\bm{h}_{\tilde{u}}}{\norm{\bm{h}_u}}x_{\tilde{u}}+
        \frac{\bm{h}_u^\dagger}{\norm{\bm{h}_u}}(\bm{z}+\tilde{\bm{z}})
        \Big|^2 \Big| \{\bm{h}_{\tilde{u}}\}\Big) 
        = 
        n_{\ell+1}^{-1}a_\ell^{\alpha/2}
        \sum_{\tilde{u}\in U(n_{\ell+1})\setminus\{u\}} 
        \tilde{r}_{\tilde{u}}^\alpha\frac{\abs{\bm{h}_u^\dagger\bm{h}_{\tilde{u}}}^2}{\norm{\bm{h}_u}^2}+
        N_0+\Delta^2,
    \end{equation*}
    where we have used the assumption that $\{z_v\}_{v\in
    V(n_{\ell+1})}$ are uncorrelated in the second line. Using
    \eqref{eq:hbounds}, this is, in turn, upper bounded by
    \begin{equation*}
        2^{3\alpha/2} \tilde{r}_u^{\alpha} n_{\ell+1}^{-2}a_\ell^{\alpha}
        \sum_{\tilde{u}\in U(n_{\ell+1})\setminus\{u\}} \tilde{r}_{\tilde{u}}^{\alpha}
        \abs{\bm{h}_u^\dagger\bm{h}_{\tilde{u}}}^2
        + N_0+\Delta^2.
    \end{equation*}
    Similarly, we can lower bound the received signal power as
    \begin{equation*}
        \E\big(\norm{\bm{h}_u}^2\abs{x_u}^2\big)
        \geq 2^{-3\alpha/2}.
    \end{equation*}
    Since Gaussian noise is the worst additive noise under a power
    constraint \cite{iha}, and applying Jensen's inequality to the
    convex function $\log(1+1/x)$, we obtain 
    \begin{align}
        \label{eq:tilderho}
        I(x_u; \hat{x}_u\vert\{h_{\tilde{u},\tilde{v}}\})
        & \geq \E \Bigg(\log\bigg(
        1+\frac{2^{-3\alpha/2}}
        {2^{3\alpha/2} \tilde{r}_u^{\alpha} n_{\ell+1}^{-2} a_\ell^{\alpha}
        \sum_{\tilde{u}\in U(n_{\ell+1})\setminus\{u\}}
        \tilde{r}_{\tilde{u}}^{\alpha} \abs{\bm{h}_u^\dagger\bm{h}_{\tilde{u}}}^2+N_0+\Delta^2}
        \bigg) \Bigg) \nonumber \\
        & \geq \log\bigg(
        1+\frac{2^{-3\alpha/2}}
        {2^{3\alpha/2} \tilde{r}_u^{\alpha} n_{\ell+1}^{-2} a_\ell^{\alpha}
        \sum_{\tilde{u}\in U(n_{\ell+1})\setminus\{u\}}
        \tilde{r}_{\tilde{u}}^{\alpha} \E\big(\abs{\bm{h}_u^\dagger\bm{h}_{\tilde{u}}}^2\big)+N_0+\Delta^2}
        \bigg).
    \end{align}
    We have for $u\neq \tilde{u}$,
    \begin{align}
        \label{eq:angle}
        \E\big(\abs{\bm{h}_u^\dagger\bm{h}_{\tilde{u}}}^2 \big)
        & = \E(\bm{h}_u^\dagger\bm{h}_{\tilde{u}}\bm{h}_{\tilde{u}}^\dagger\bm{h}_u) \nonumber\\
        & = \sum_{v\in V(n_{\ell+1})} \abs{h_{u,v}}^2\abs{h_{\tilde{u},v}}^2 \nonumber\\
        & = \sum_{v\in V(n_{\ell+1})} r_{u,v}^{-\alpha} r_{\tilde{u},v}^{-\alpha},
    \end{align}
    and hence using \eqref{eq:rtildebounds}
    \begin{align*}
        \E\bigg(\tilde{r}_u^\alpha \sum_{\tilde{u} \in U(n_{\ell+1})\setminus\{u\}} 
        \tilde{r}_{\tilde{u}}^\alpha \abs{\bm{h}_u^\dagger\bm{h}_{\tilde{u}}}^2 \bigg)
        & \qquad = \sum_{\tilde{u} \in U(n_{\ell+1})\setminus\{u\}} \sum_{v\in V(n_{\ell+1})} 
        \tilde{r}_u^{\alpha}r_{u,v}^{-\alpha}
        \tilde{r}_{\tilde{u}}^{\alpha}r_{\tilde{u},v}^{-\alpha} \\
        & \qquad \leq 2^{-\alpha} n_{\ell+1}^2 a_{\ell}^{-\alpha}.
    \end{align*}
    Therefore we can continue~\eqref{eq:tilderho} as
    \begin{equation}
        \label{eq:i2}
        I(x_u; \hat{x}_u\vert\{h_{\tilde{u},\tilde{v}}\}) 
        \geq \frac{1}{2}\log\bigg(1+\frac{2^{-3\alpha/2}}{2^{\alpha/2}+N_0+\Delta^2}\bigg)
        \defeq K_4.
    \end{equation}

    Using~\eqref{eq:i1} and~\eqref{eq:i2} in Lemma~\ref{thm:quant}, and
    observing that we only communicate during a fraction
    \begin{equation*}
        P_\ell(n) n_{\ell+1}a_\ell^{-\alpha/2}\leq 1
    \end{equation*}
    of time yields a per source node rate $\rho^\tsf{MAC}_\ell(n)$ arbitrarily
    close to 
    \begin{equation*}
        K_4 P_\ell(n) n_{\ell+1}a_\ell^{-\alpha/2}
    \end{equation*}
    and a quantizer of rate arbitrarily close to
    \begin{equation*}
        \log\Big(1+\frac{2^{-\alpha/2}+N_0}{\Delta^2}\Big)
    \end{equation*}
    bits per observation at each relay node. Since by~\eqref{eq:i2}
    mutual information
    $I(x_u;\hat{x}_u\vert\{h_{\tilde{u},\tilde{v}}\})$ is at least $K_4$
    for every $u\in U(n_{\ell+1})$ during the fraction of time we
    actually communicate, this implies that there are at most $1/K_4$
    observations at each relay node per $n_{\ell+1}$ total message bits.
    Thus the number of bits per relay node required to quantize the
    observations is at most 
    \begin{equation*}
        K_5 \defeq
        \frac{1}{K_4}\log\Big(1+\frac{2^{-\alpha/2}+N_0}{\Delta^2}\Big)
    \end{equation*}
    bits per $n_{\ell+1}$ total message bits sent by the source nodes.
\end{proof}

\subsection{Broadcast Phase}
\label{sec:bc}

At the end of the MAC phase, each node in the relay squarelet received a
part of the message sent by each source node. In the BC phase, each node
in the relay squarelet encodes these messages together for $n_{\ell+1}$
transmit antennas. The encoded message is then quantized and
communicated to all the nodes in the relay squarelet. These nodes then
send the quantized encoded message to the destination nodes
$W(n_{\ell+1})$.  Note that this again induces a uniform traffic pattern
between the nodes in the relay squarelet, i.e., every node needs to
transmit quantized encoded messages to every other node. While this
traffic pattern does not correspond to a permutation traffic matrix it
can be written as a sum of $n_{\ell+1}$ permutation traffic matrices. A
fraction $1/n_{\ell+1}$ of the traffic within the relay squarelet is
transmitted according to each of these permutation traffic matrices.
This setup is depicted in Figure~\ref{fig:bc} in
Section~\ref{sec:construction}.

Assuming for the moment that we have a scheme to send the quantized
encoded messages to the corresponding nodes in the relay squarelet, the
traffic matrix $\widetilde{S}^{(1,1)}(n_{\ell+1})$ between
$V(n_{\ell+1})$ and $W(n_{\ell+1})$ describes then a BC with one
transmitter with $n_{\ell+1}$ antennas and $n_{\ell+1}$ receivers, each
with one antenna. We call this the \emph{BC induced by
$\widetilde{S}^{(1,1)}(n_{\ell+1})$} in the following.

\begin{lemma}
    \label{thm:bc}
    For the BC induced by $\widetilde{S}^{(1,1)}(n_{\ell+1})$ with per-node
    average power constraint $P_\ell(n)\leq
    n_{\ell+1}^{-1}a_\ell^{\alpha/2}$, a rate of
    \begin{equation*}
        \rho^\textup{\tsf{BC}}_\ell(n) \geq K_6
        P_\ell(n)n_{\ell+1}a_\ell^{-\alpha/2}
    \end{equation*}
    is achievable per destination node, and the number of bits required
    to quantize the observations is at most $K_7 (\ell+1) \log(n)$ bits
    at each relay node per $n_{\ell+1}$ total message
    bits\footnote{Total message bits refers to the sum of all message
    bits received by the $n_{\ell+1}$ destination nodes.} received by
    the destination nodes. 
\end{lemma}
\begin{proof}
    Consider a node $v\in V(n_{\ell+1})$ in the relay squarelet, say the
    first one. From the MAC phase, this node received the first part of
    the messages of each source node $u\in U(n_{\ell+1})$. We would like
    to jointly encode these message parts at the relay node using
    transmit beamforming, and then transmit the corresponding encoded
    signal using all the nodes in the relay squarelet. However, this
    cannot be done directly, because at the encoding time, the future
    channel state at transmission time is unknown. 
    
    We circumvent this problem by reordering the signals to be
    transmitted at the relay nodes as follows. Let
    \begin{equation*}
        \{\hat{\theta}_{v,w}\}_{v\in V(n_{\ell+1}), w\in W(n_{\ell+1})}
        \in\{0,\pi/2,\pi,3\pi/2\}^{n_{\ell+1}^2}.
    \end{equation*}
    be a ``quantized'' channel state. The part of the messages at
    node one in the relay squarelet is encoded for $n_{\ell+1}$ transmit nodes
    with an assumed channel gain of 
    \begin{equation*}
        \hat{h}_{v,w}[t] =
        r_{v,w}^{-\alpha/2}\exp(\sqrt{-1}\hat{\theta}_{v,w}[t]),
    \end{equation*}
    where the $\{\hat{\theta}_{v,w}[t]\}_{v,w,t}$ are cycled as a
    function of $t$ through all possible values in
    $\{0,\pi/2,\pi,3\pi/2\}^{n_{\ell+1}^2}$. The components of the
    encoded messages are then quantized and each component sent to the
    corresponding node in the relay squarelet. Once all nodes in the
    relay squarelet have received the encoded message, they send in each
    time slot a sample of the encoded messages corresponding to the
    quantized channel state closest (in Euclidean distance) to the
    actual channel realization in that time slot. By ergodicity of
    $\{\theta_{u,v}[t]\}_{t}$, each quantized channel state is used
    approximately the same number of times. More precisely, as the
    message length grows to infinity, we can send samples of the encoded
    message parts a $1/(1+\eta)$ fraction of time with probability
    approaching $1$ for any $\eta>0$.  Since we have no constraint on
    the encoding delay in our setup, we can choose $\eta$ arbitrarily
    small, and given that we are only interested in scaling laws, we
    will ignore this term in the following to simplify notation. Note
    that the destination nodes can reorder the received samples since we
    assume full CSI. In the following, we let
    $\{\hat{\theta}_{v,w}\}_{v,w}$ be the random quantized channel state
    induced by $\{\theta_{v,w}\}_{v,w}$ through the above procedure.
    Denote by $\{\hat{h}_{v,w}\}_{v,w}$ the corresponding channel
    gains.

    As in the MAC phase, the nodes in the relay squarelet send signals
    at a power (essentially) $n_{\ell+1}^{-1}a_\ell^{\alpha/2}$ a
    fraction $P_\ell(n)n_{\ell+1}a_\ell^{-\alpha/2}\leq 1$ of time and
    are silent for the remaining time. To create interference at uniform
    power, this is done in the same randomized manner as in the MAC
    phase.  Generate independently for each region $A(a_\ell)$ a
    Bernoulli process $\{B[t]\}_{t\in\mbb{N}}$ with parameter
    $P_\ell(n)n_{\ell+1}a_\ell^{-\alpha/2}/(1+\eta)$ for some small
    $\eta>0$. The nodes in $A(a_\ell)$ are active whenever $B[t]=1$ and
    remain silent otherwise. As before, we ignore the additional
    $1/(1+\eta)$ term. Again we only need to consider the fraction of
    time during which $B[t] = 1$.  
    
    Consider the message part at a relay node for destination node $w\in
    W(n_{\ell+1})$.  We encode this part independently; call
    $\tilde{x}_w$ the encoded message part. The relay node then performs
    transmit beamforming to construct the encoded message for all its
    destination nodes, i.e., 
    \begin{equation*}
        \bm{x} = \sum_{w\in W(n_{\ell+1})}\frac{\hat{\bm{h}}_w^\dagger}{\norm{\bm{h}_w}}
        \tilde{x}_w,
    \end{equation*}
    where row vector $\bm{h}_w=\{h_{v,w}\}_{v\in V(n_{\ell+1})}$ contains the
    channel gains to node $w$, and where we have used
    $\lvert\hat{h}_{v,w}\rvert=\abs{h_{v,w}}$. The relay node then quantizes
    the vector of encoded messages componentwise and forwards the
    quantized version $\hat{\bm{x}}$ to the other nodes in the relay
    squarelet. These nodes then send $\hat{\bm{x}}$ over the channel to
    the destination nodes. The received signal at destination node $w$
    is thus
    \begin{equation*}
        y_w = \bm{h}_w\hat{\bm{x}}+z_w.
    \end{equation*}

    With this, we have the setup considered in Lemma~\ref{thm:quant}
    (with different variable names). The first ``channel'' in
    Lemma~\ref{thm:quant} (see Figure~\ref{fig:quant}) will correspond
    to the transmit beamforming used at the relay squarelet.  The second
    channel in Lemma~\ref{thm:quant} will now correspond to the wireless
    channel between the relay squarelet $V(n_{\ell+1})$ and a
    destination node $w$. To apply Lemma~\ref{thm:quant}, we need to
    find a distribution for $\tilde{x}_w$ and for $\hat{x}_v|x_v$. We
    also need to guarantee that $\hat{x}_v$ satisfies the power
    constraint at each node $v$ in the relay squarelet. For each $w\in
    W(n_{\ell+1})$ let
    $\tilde{x}_w\sim\mc{N}_{\mbb{C}}(0,Kn_{\ell+1}^{-1}a_\ell^{\alpha/2})$
    (for some $K$ to be chosen later) independent of
    $\tilde{x}_{\tilde{w}}$ for $w\neq\tilde{w}$, and let
    $\hat{x}_v=x_v+\tilde{z}_v$ for
    $\tilde{z}_v\sim\mc{N}_{\mbb{C}}(0,\Delta^2)$ independent of
    $\bm{x}$ and for some $\Delta^2>0$. We then have
    \begin{equation*}
        y_w =  \frac{\bm{h}_w\hat{\bm{h}}_{w}^\dagger}{\norm{\bm{h}_w}}\tilde{x}_w
        +\sum_{\tilde{w}\in W(n_{\ell+1})\setminus\{w\}}
        \frac{\bm{h}_w\hat{\bm{h}}_{\tilde{w}}^\dagger}{\norm{\bm{h}_{\tilde{w}}}}
        \tilde{x}_{\tilde{w}}+\bm{h}_w\tilde{\bm{z}}+z_w.
    \end{equation*}

    We proceed by computing the mutual informations $I(x_v;
    \hat{x}_v\vert\{h_{\tilde{u},\tilde{v}}\})$ and $I(\tilde{x}_w;
    y_w\vert\{h_{\tilde{u},\tilde{v}}\})$ as required in Lemma
    \ref{thm:quant} (the conditioning in $\{h_{\tilde{u},\tilde{v}}\}$
    again being due to the availability of full CSI).  Note first that
    by construction of $\widetilde{S}^{(1,1)}(n_{\ell+1})$, we have for
    any $w\in W(n_{\ell+1})$ 
    \begin{equation*}
        2 \min_{v\in V(n_{\ell+1})} r_{v,w} 
        \geq \max_{v\in V(n_{\ell+1})} r_{v,w},
    \end{equation*}
    and therefore
    \begin{equation}
        \label{eq:hbound2}
        \frac{\abs{h_{v,w}}^2}{\norm{\bm{h}_{w}}^2}
        \leq \frac{\big(\min_{v\in V(n_{\ell+1})} r_{v,w}\big)^{-\alpha}}
        {n_{\ell+1}\big(\max_{v\in V(n_{\ell+1})} r_{v,w}\big)^{-\alpha}} \\
        \leq \frac{2^\alpha}{n_{\ell+1}}.
    \end{equation}

    We start by computing $I(x_v;
    \hat{x}_v\vert\{h_{\tilde{u},\tilde{v}}\})$. $\hat{x}_v$ has mean
    zero and variance 
    \begin{align}
        \label{eq:xhatvar}
        \E \big(\abs{\hat{x}_v}^2 \big)
        & = \sum_{w\in W(n_{\ell+1})}\frac{\abs{h_{v,w}}^2}{\norm{\bm{h}_{w}}^2}
        Kn_{\ell+1}^{-1}a_\ell^{\alpha/2} +\Delta^2 \nonumber\\
        & \leq n_{\ell+1} \frac{2^\alpha}{n_{\ell+1}}
        Kn_{\ell+1}^{-1}a_\ell^{\alpha/2}+\Delta^2 \nonumber\\
        & \leq n_{\ell+1}^{-1}a_\ell^{\alpha/2},
    \end{align}
    for
    \begin{equation*}
        K \defeq 2^{-\alpha}(1-\Delta^2),
    \end{equation*}
    which is positive for $\Delta^2<1$, and where we have used
    \eqref{eq:hbound2} and that
    \begin{equation*}
        n_{\ell+1}^{-1}a_{\ell}^{\alpha/2}
        \geq 2^{\ell+1}\gamma(n)
        \geq 1
    \end{equation*}
    by \eqref{eq:assumptions}.  Equation~\eqref{eq:xhatvar} shows that
    $\hat{x}_v$ satisfies the power constraint of node $v$ in the relay
    squarelet $V(n_{\ell+1})$.  Moreover, we obtain
    \begin{align}
        \label{eq:i3}
        I(x_v; \hat{x}_v\vert\{h_{\tilde{u},\tilde{v}}\})
        & = h(\hat{x}_v\vert\{h_{\tilde{u},\tilde{v}}\})
        -h(\hat{x}_v|x_v,\{h_{\tilde{u},\tilde{v}}\}) \nonumber\\
        & \leq \log\Big(2\pi e \E\big(\abs{\hat{x}_v}^2\big)\Big)-\log(2\pi e\Delta^2) \nonumber\\
        & \leq \log\bigg(\frac{n_{\ell+1}^{-1}a_\ell^{\alpha/2}}{\Delta^2}\bigg).
    \end{align}

    It remains to compute $I(\tilde{x}_w;
    y_w\vert\{h_{\tilde{u},\tilde{v}}\})$. Note that the encoding
    procedure guarantees that
    \begin{equation*}
        \cos(\pi/4)^2\norm{\bm{h}_w}^4
        \leq \lvert\bm{h}_w\hat{\bm{h}}^\dagger_w\rvert^2
        \leq \norm{\bm{h}_w}^4.
    \end{equation*}
    Moreover, for $w\neq\tilde{w}$,
    \begin{align*}
        \E\big(\lvert \bm{h}_{w}\hat{\bm{h}}^\dagger_{\tilde{w}}\rvert^2\big)
        & = \E (\bm{h}_{w}\hat{\bm{h}}^\dagger_{\tilde{w}}\hat{\bm{h}}_{\tilde{w}} \bm{h}^\dagger_{w}) \\
        & = \sum_{v\in V(n_{\ell+1})} 
        \E \big(\abs{h_{vw}}^2\lvert\hat{h}_{v\tilde{w}}\rvert^2\big) \\
        & = \sum_{v\in V(n_{\ell+1})} 
        \E \big(\abs{h_{vw}}^2\abs{h_{v\tilde{w}}}^2\big) \\
        & = \E\big( \lvert \bm{h}_{w}\bm{h}^\dagger_{\tilde{w}}\rvert^2\big) .
    \end{align*}
    From this, we get by a similar argument as in Lemma~\ref{thm:mac} that
    \begin{align}
        \label{eq:i4}
        I(\tilde{x}_w; y_w\vert\{h_{\tilde{u},\tilde{v}}\}) 
        \geq K_6.
    \end{align}

    Using~\eqref{eq:i3} and~\eqref{eq:i4} in Lemma~\ref{thm:quant}, and
    observing that we only communicate during a fraction
    \begin{equation*}
        P_\ell(n)n_{\ell+1}a_\ell^{-\alpha/2}
    \end{equation*}
    of time, yields a per destination node rate $\rho^\tsf{BC}_\ell(n)$
    arbitrarily close to 
    \begin{equation*}
        K_6 P_\ell(n)n_{\ell+1}a_\ell^{-\alpha/2}
    \end{equation*}
    bits per channel use and a quantizer rate arbitrarily close to
    \begin{equation*}
        \log\Big(\frac{n_{\ell+1}^{-1}a_\ell^{\alpha/2}}{\Delta^2}\Big)
    \end{equation*}
    bits per encoded sample. Since by~\eqref{eq:i4} mutual information
    $I(\tilde{x}_w; y_w\vert\{h_{\tilde{u},\tilde{v}}\})$ is at least
    $K_6$ for every $w\in W(n_{\ell+1})$ during the fraction of time we
    actually communicate, this implies that there are at most $1/K_6$
    encoded message samples for each relay node per $n_{\ell+1}$ total
    message bits received by the destination nodes $W(n_{\ell+1})$. Thus
    the number of bits required at each relay node to quantize the
    encoded message samples is at most 
    \begin{align*}
        \frac{1}{K_6}\log\Big(\frac{n_{\ell+1}^{-1}a_\ell^{\alpha/2}}{\Delta^2}\Big)
        & = \frac{1}{K_6}
        \log\Big(\frac{1}{\Delta^2}2^{\ell+1}\gamma^{1+\ell(1-\alpha/2)}(n) n^{\alpha/2-1}\Big) \\
        & \leq \frac{1}{K_6}
        \log\Big(\frac{1}{\Delta^2}2^{\ell+1} n^{\alpha/2}\Big) \\
        & \leq K_7 (\ell+1) \log(n)
    \end{align*}
    bits per $n_{\ell+1}$ total message bits received by the destination
    nodes, and where we have used $\gamma(n)\leq n$ by
    \eqref{eq:assumptions}.
\end{proof}

\section{Proof of Theorem~\ref{thm:achievability}}
\label{sec:achievability}

The proof of Theorem \ref{thm:achievability} is split into two parts. In
Section \ref{sec:fast} we prove the theorem for fast fading, and in
Section \ref{sec:slow} for slow fading.

\subsection{Fast Fading}
\label{sec:fast}

In this section, we prove Theorem \ref{thm:achievability} under fast
fading, i.e., $\{\theta_{u,v}[t]\}_{t}$ is stationary and ergodic in
$t$. We first prove that the assumptions on the power constraint and the
interference made in Section \ref{sec:analysis} (see Lemmas
\ref{thm:mac} and \ref{thm:bc}) during the analysis of one level of the
hierarchical relaying scheme are valid. We then use the results proved
there to analyze the behavior of the entire hierarchy, yielding a lower
bound on the per-node rate achievable with hierarchical relaying.

We first argue that the constraint $P_\ell(n) \leq
n_{\ell+1}^{-1}a_\ell^{\alpha/2}$ needed in Lemmas \ref{thm:mac} and
\ref{thm:bc} is satisfied.  Consider the hierarchical relaying scheme as
described in Section~\ref{sec:desc} and fix a level $\ell$, $0\leq\ell<
L=L(n)$ in this hierarchy. At level $\ell$ we have a square of area
$a_\ell= n/\gamma^{\ell}(n)$, with $n_\ell= n/2^\ell \gamma^\ell(n)$
source-destination pairs.  Since we are time sharing between
$K_22^{-\ell}\gamma(n)$ relay squarelets at this level, we have an
average power constraint of 
\begin{equation*}
    P_\ell(n)\defeq K_2 2^{-\ell}\gamma(n)
\end{equation*}
during the time any particular relay squarelet is active.
Since $\alpha>2$ and since $n\gamma^{-L(n)}(n)\to\infty$ as
$n\to\infty$, we have, for $n$ large enough (independent of $\ell$),
that
\begin{align*}
    P_\ell(n)
    & = K_22^{-\ell}\gamma(n) \\
    & \leq 2^{-\ell} \gamma(n) \Big(\frac{n}{\gamma^{L(n)}(n)}\Big)^{\alpha/2-1} \\
    & \leq 2^{\ell+1} \gamma(n) \Big(\frac{n}{\gamma^{\ell}(n)}\Big)^{\alpha/2-1} \\
    & = n_{\ell+1}^{-1}a_\ell^{\alpha/2}.
\end{align*}
Therefore the power constraint in Lemmas~\ref{thm:mac}
and~\ref{thm:bc} is satisfied. 

We continue by analyzing the interference caused by spatial re-use.
Recall that the MAC and BC phases at level $\ell$ induce permutation
traffic within the dense squarelets at level $\ell+1$. The permutation
traffic within those dense squarelets at level $\ell+1$ is transmitted
in parallel with spatial re-use. We now describe in detail how this
spatial re-use is performed. Partition the squarelets of area
$a_{\ell+1}$ (i.e., at level $\ell+1$) into four subsets such that in
each subset all squarelets are at distance at least $\sqrt{a_{\ell+1}}$
from each other. The traffic that the MAC and BC phases at level $\ell$
induce in each of the relay squarelets at level $\ell+1$ is transmitted
simultaneously within all relay squarelets in the same subset. Consider
now one such subset.  We show that at any relay squarelet the
interference from other relay squarelets in the same subset is
stationary and ergodic within each phase, additive (i.e., independent of
the signals and channel gains in this relay squarelet), and of bounded
power $N_0-1$ independent of $n$. 

We first argue that the interference is stationary and ergodic within
each phase. Note first that on any level $\ell+1$ in the hierarchy, all
relay squarelets are either simultaneously in the MAC phase or
simultaneously in the BC phase. Furthermore, all relay squarelets are
also synchronized for transmissions within each of these phases (recall
that the induced traffic in level $\ell+1$ is uniform and is sent
sequentially as permutation traffic). Hence it suffices to show that the
interference generated by either the MAC or BC induced by some
permutation traffic matrix is stationary and ergodic. Since all
codebooks for either of these cases are generated as i.i.d. Gaussian
multiplied by a Bernoulli process, and in the BC phase beamformed for
stationary and ergodic fading, this is indeed the case.

The additivity of the interference follows easily for the MAC phase,
since codebooks are generated independently of the channel realization
in this case. Moreover, since the channel gains are independent from
each other and all codebooks are generated as independent zero mean
processes, the interference in the MAC phase is also uncorrelated (over
space) within each relay squarelet. For the BC phase, the codebook
depends only on the channel gains within each relay squarelet at level
$\ell+1$. Since the channel gains within relay squarelets are
independent of the channel gains between relay squarelets, this
interference is additive as well. 

We now bound the interference power. Note that by the randomized
time-sharing construction within the MAC and BC phases (see Lemmas
\ref{thm:mac} and \ref{thm:bc}), in each relay squarelet, at most
$n_{\ell+1}$ nodes transmit at an average power of $1$. In the MAC
phase, all nodes use independently generated codebooks with power at
most $1$, and thus the received interference power from another relay
squarelet at distance $i\sqrt{a_{\ell+1}}$ is at most
\begin{equation*}
    n_{\ell+1} i^{-\alpha}a_{\ell+1}^{-\alpha/2}
    = i^{-\alpha}2^{-(\ell+1)}\Big(\frac{n}{\gamma^{\ell+1}(n)}\Big)^{1-\alpha/2}
    \leq i^{-\alpha},
\end{equation*}
by~\eqref{eq:assumptions}.  In the BC phase, the nodes in each active
relay squarelet use beamforming to transmit to nodes within their own
squarelet.  Since the channel gains within a relay squarelet are
independent of the channel gains between relay squarelets, the same
calculation as in \eqref{eq:angle} shows that we can upper bound the
received interference power from another relay squarelet at distance
$i\sqrt{a_{\ell+1}}$ by
\begin{equation*}
    n_{\ell+1} i^{-\alpha}a_{\ell+1}^{-\alpha/2}
    \leq i^{-\alpha},
\end{equation*}
in the BC phase as well. 

Now, by the way in which we perform spatial re-use, every active relay
squarelet has at most $8i$ active relay squarelets at distance at least
$i\sqrt{a_{\ell+1}}$. Hence the total interference power received at an
active relay squarelet is at most
\begin{equation*}
    \sum_{i=1}^{\infty}8i2^\alpha i^{-\alpha} \defeq N_0-1 < \infty
\end{equation*}
since $\alpha > 2$. With this, we have shown that the interference term
has the properties required for Lemmas~\ref{thm:mac} and~\ref{thm:bc} to
apply.

We now apply those two lemmas to obtain a lower bound on the rate
achievable with hierarchical relaying. Call $\tau_\ell(n)$ the number of
channel uses to transmit one bit from each of $n_{\ell}$ source nodes to
the corresponding destination nodes at level $\ell$.
Lemma~\ref{thm:relay} states that for $n$ large enough (independent of
$\ell$), we relay over each dense squarelet at most $K_3 2^\ell$ times.
Combining this with Lemma~\ref{thm:mac}, we see that to transmit one bit
from each source to its destination at this level we need at most
\begin{equation*}
    4K_32^\ell K_22^{-\ell}\gamma(n) \frac{1}{K_4P_\ell(n)}
    n_{\ell+1}^{-1}a_{\ell}^{\alpha/2} \\
    = \frac{K_3 2^{2\ell+3}}{K_4}n^{\alpha/2-1}\gamma^{1+\ell(1-\alpha/2)}(n)
\end{equation*}
channel uses for the MAC phase. Here, the factor $4$ accounts for
the spatial re-use, $K_32^\ell$ accounts for relaying over the same
relay squarelets multiple times, $K_22^{-\ell}\gamma(n)$ accounts for
time sharing between the relay squarelets, and the last term
accounts for the time required to communicate over the MAC.
Similarly, combining Lemmas~\ref{thm:relay} and~\ref{thm:bc}, we need
at most
\begin{equation*}
    \frac{K_3 2^{2\ell+3}}{K_6}n^{\alpha/2-1}\gamma^{1+\ell(1-\alpha/2)}(n)
\end{equation*}
channel uses for the BC phase. Moreover, at level $\ell+1$ in the hierarchy 
this induces a per-node traffic demand of at most $K_5$ bits from the MAC phase, and at
most $K_7(\ell+1)\log(n)$ from the BC phase. Thus we obtain the following
recursion 
\begin{align}
    \label{eq:rec}
    \tau_\ell(n) 
    & \leq 8 K_3\Big(\frac{1}{K_4} + \frac{1}{K_6}\Big) 
    n^{\alpha/2-1}\gamma(n)\big(4\gamma^{1-\alpha/2}(n)\big)^{\ell}
    + (K_5+K_7(\ell+1)\log(n))\tau_{\ell+1}(n) \nonumber \\
    & \leq \widetilde{K}n^{\alpha/2-1}\gamma(n)4^{\ell}+
    K(\ell+1)\log(n)\tau_{\ell+1}(n) \nonumber\\
    & \leq \widetilde{K}n^{\alpha/2-1}\gamma(n)4^{L}+
    KL\log(n)\tau_{\ell+1}(n)
\end{align}
for positive constants $K,\widetilde{K}$ independent of $n$ and
$\ell$.

We use TDMA at scale $a_L$ with $n_L$ nodes and
source-destination pairs. Time sharing between all source-destination
pairs, we have (during the time we communicate for each node) an average
power constraint of $n_L$. Since at this level we communicate over
a distance of at most $2a_L^{1/2}$, we have 
\begin{equation}
    \label{eq:tauL}
    \tau_L(n) 
    \leq n_L\log^{-1}\bigg(1+\frac{n_L}{2^{\alpha}N_0 a_L^{\alpha/2}}\bigg).
\end{equation}
Since
\begin{equation*}
    n_L a_L^{-\alpha/2} 
    \leq n_L a_L^{-1} 
    = 2^{-L}
    \to 0
\end{equation*}
as $n\to\infty$, we can upper bound
\eqref{eq:tauL} as
\begin{align}
    \label{eq:tauL2}
    \tau_L(n) 
    & \leq K' a_L^{\alpha/2} \nonumber\\
    & = K' n^{\alpha/2}\gamma^{-L \alpha/2}(n) \nonumber\\
    & \leq K' n^{\alpha/2}\gamma^{-L}(n)
\end{align}
for some constant $K'$. 

Now, using the recursion~\eqref{eq:rec} $L$ times, and combining with
\eqref{eq:tauL2}, we obtain
\begin{align}
    \label{eq:t0}
    \tau_0 (n)
    & \leq \widetilde{K}n^{\alpha/2-1}\gamma(n)4^L+ KL\log(n)\tau_1(n) \nonumber\\
    & \leq \ldots \nonumber\\
    & \leq \widetilde{K}n^{\alpha/2-1}\gamma(n)4^L\bigg(\sum_{\ell=0}^{L-1}\big(K L \log(n)\big)^{\ell}\bigg) \nonumber\\
    & \quad + \big(KL\log(n)\big)^L \tau_L(n) \nonumber\\
    & \leq n^{\alpha/2-1}\big( K L \log(n)\big)^{L}
    \Big(\widetilde{K}4^L\gamma(n) + K' n\gamma^{-L}(n)\Big).
\end{align}
Using the definition of $\gamma(n)$ and $L=L(n)$ in \eqref{eq:gammadef},
we have for $n$ large enough
\begin{align*}
    \big(KL(n)\log(n)\big)^{L(n)} 
    & \leq n^{2\log^{-1/2-\delta}(n)\log\log(n)}, \\
    4^{L(n)}\gamma(n)
    & \leq n^{2\log^{-1/2-\delta}(n)+\log^{\delta-1/2}(n)}, \\
    n\gamma^{-L(n)}(n)
    & \leq n^{\log^{\delta-1/2}(n)}.
\end{align*}
Since $\delta>0$, the $n^{\log^{\delta-1/2}(n)}$ term dominates in 
\eqref{eq:t0}, and we obtain 
\begin{equation*}
    \tau_0(n)
    \leq \tilde{b}(n)n^{\alpha/2-1},
\end{equation*}
where 
\begin{equation*}
    \tilde{b}(n) \leq n^{O(\log^{\delta-1/2}(n))},
\end{equation*}
as $n\to\infty$. Therefore
\begin{equation*}
    \rho^*(n) 
    \geq \rho^{\tsf{HR}}(n)
    = 1/\tau_0(n) 
    \geq b(n)n^{1-\alpha/2},
\end{equation*}
with
\begin{equation*}
    b(n) \geq n^{-O(\log^{\delta-1/2}(n))},
\end{equation*}
concluding the proof for the fast fading case.

\subsection{Slow Fading}
\label{sec:slow}

In this section, we prove Theorem \ref{thm:achievability} under slow
fading, i.e., $\{\theta_{u,v}[t]\}_{t}$ is constant as a function of
$t$. We sketch the necessary modifications for the scheme described in
Section~\ref{sec:desc} to achieve a per-node rate of at least
$b(n)n^{1-\alpha/2}$ in the slow fading case. 

Consider level $\ell$, $0\leq \ell < L(n)$ in the hierarchy.
Instead of relaying the message of a source-destination pair over
one relay squarelet as in the scheme described in
Section~\ref{sec:desc}, we relay the message over many dense
squarelets that are at least at distance $\sqrt{2a_{\ell+1}}$ from
both the source and the destination nodes. We time share between the
different relays. The idea here is that the wireless channel between
any node and its relay squarelet might be in a bad state due to the
slow fading, making communication over this relay squarelet
impossible. Averaged over many relay squarelets, however, we get
essentially the same performance as in the fast fading case.

We first state a (somewhat weaker) version of Lemma \ref{thm:relay},
appropriate for this setup. Consider again the collection of schedules
$\mc{S}(n_\ell)$ and $\widetilde{\mc{S}}(n_\ell)$ satisfying the
conditions that no relay squarelet is selected by more than $n_{\ell+1}$
source-destination pairs and that all sources and destinations are at
least at distance $\sqrt{2a_{\ell+1}}$ from their relay squarelet (see
Section \ref{sec:relay} for the formal definition). The next lemma shows
that for each source-destination pair, we can find
$K_22^{-\ell-1}\gamma(n)$ distinct relay squarelets satisfying the above
conditions (the requirement that these relay squarelets are distinct is
expressed by the orthogonality condition of the schedules in Lemma
\ref{thm:relay1} below).
\begin{lemma}
    \label{thm:relay1} 
    For every $n$ large enough (independent of $\ell$) and every
    permutation traffic matrix
    $\lambdauc(n_\ell)\in\{0,1\}^{n_\ell\times n_\ell}$ there are schedules
    $\{S^{(i)}(n_\ell)\}_{i=1}^{K_2 2^{-\ell} \gamma^2(n)} \subset
    \mc{S}(n_\ell)$, ${\{\widetilde{S}^{(i)}(n_\ell)\}}_{i=1}^{K_2
    2^{-\ell} \gamma^2(n)} \subset \widetilde{\mc{S}}(n_\ell)$
    satisfying
    \begin{equation*}
        \lambdauc(n_\ell) = \frac{1}{K_22^{-\ell-1}\gamma(n)} 
        \sum_{i=1}^{ K_2 2^{-\ell}\gamma^2(n)} S^{(i)}(n_\ell)\widetilde{S}^{(i)}(n_\ell),
    \end{equation*} 
    where $\{S^{(i)}(n_\ell)\}_i$, $\{\widetilde{S}^{(i)}(n_\ell)\}_i$ 
    are collections of orthogonal matrices in the sense that for $i\neq i'$,
    \begin{equation}
        \label{eq:disjoint}
        \begin{aligned}
            \sum_{u,k} s_{u,k}^{(i)} s_{u,k}^{(i')} & = 0, \\
            \sum_{k,u} \tilde{s}_{k,u}^{(i)}
            \tilde{s}_{k,u}^{(i')} & = 0.
        \end{aligned}
    \end{equation}
\end{lemma}
\begin{proof}
    The proof is similar to that of Lemma \ref{thm:relay}. In order to
    construct $\{S^{(i)}(n_\ell)\}$ and
    $\{\widetilde{S}^{(i)}(n_\ell)\}$, consider the sequential pass over
    all $n$ source-destination pairs (assume $n$ is large enough for
    Lemma \ref{thm:relay} to hold). As before, for each
    source-destination pair, there are $K_22^{-\ell-1}\gamma(n)$ dense
    relay squarelets that are at distance at least $\sqrt{2a_{\ell+1}}$.
    Each pair chooses all of these $K_22^{-\ell-1}\gamma(n)$ squarelets,
    instead of just one as before. Stop one round of this procedure as
    soon as any of the relay squarelets is chosen by $n_{\ell+1}$ pairs.
    Since by the end of one round at least one relay squarelet is
    matched by $n_{\ell+1}$ source-destination pairs, there are at most
    $n_{\ell}/n_{\ell+1}=2\gamma(n)$ such rounds.

    Consider now the result of one such round. We construct $K_2
    2^{-\ell-1}\gamma(n)$ matrices $S^{(i)}(n_\ell)$ and
    $\widetilde{S}^{(i)}(n_\ell)$, with the $i$-th pair of matrices
    describing communication over the $i$-th relay squarelets chosen
    by source-destination pairs matched in this round. Thus, this
    process produces a total of $2\gamma(n) K_2 2^{-\ell-1}\gamma(n)=K_2
    2^{-\ell} \gamma^2(n)$ such matrices. The orthogonality property
    follows since each source-destination pair relays over the same
    relay squarelet only once.
\end{proof}

Given a decomposition of the scaled traffic matrix
$K_22^{-\ell-1}\gamma(n)\lambdauc(n)$ into $K_2 2^{-\ell}\gamma^2(n)$ matrices,
each source-destination pair tries to relay over $K_2
2^{-\ell-1}\gamma(n)$ dense squarelets. We time share between these
relay squarelets. Since each source-destination pair relays only a 
$(K_2 2^{-\ell-1}\gamma(n))^{-1}$ fraction of traffic over any of its
relay squarelets, the loss due to this time sharing is now
\begin{equation*}
    \frac{K_2 2^{-\ell}\gamma^2(n)}{K_2 2^{-\ell-1}\gamma(n)} 
    = 2\gamma(n)
\end{equation*}
as opposed to $K_32^{\ell}$ in Lemma~\ref{thm:relay}.  In other words,
the loss is at most a factor $2\gamma(n)$ more than in
Lemma~\ref{thm:relay}.  Using the definition of $\gamma(n)$
in~\eqref{eq:gammadef}, we have 
\begin{equation*}
    \gamma(n)\leq n^{-\log^{\delta-1/2}(n)}\leq b^{-1}(n). 
\end{equation*}
In other words, this additional loss is small.

Consider now a specific relay squarelet. If a source-destination pair
can communicate over this relay squarelet at a rate at least 
$1/64$-th of the rate achievable in the fast fading case (given by
Lemmas~\ref{thm:mac} and~\ref{thm:bc}), it sends information over this
relay. Otherwise it does not send anything during the period of time it
is assigned this relay. We now show that, with probability $1-o(1)$ as
$n\to\infty$, for every source-destination pair on every level of the
hierarchy at least one quarter of its relay squarelets can support this
rate.  As we only communicate over a quarter of the relay squarelets,
this implies that we can achieve at least $1/256$-th of the
per-node rate for the fast fading case (see Section \ref{sec:fast}),
i.e., that $b(n)n^{1-\alpha/2}$ is achievable with probability
$1-o(1)$ as $n\to\infty$. 

Assume we have for each source-destination pair $(u,w)$ picked
$K_22^{-\ell-1}\gamma(n)$ dense squarelets over which it can relay; call
those relay squarelets $\{A_{u,w,k}\}_{k=1}^{K_22^{-\ell-1}\gamma(n)}$.
Consider the event $B_{u,w,k}$ that source node $u$ can communicate
at the desired rate to destination node $w$ over relay squarelets
$A_{u,w,k}$ (assuming, as before, that we can solve the
communication problem within this squarelet). 

Let $\{B_{u,w,k}^{(i)}\}_{i=1}^4$ be the events that the interference
due to matched filtering in the MAC phase, the interference from spatial
re-use in the MAC phase, the interference due to beamforming in the BC
phase, and the interference from spatial re-use in the BC phase, are less
than $8$ times the one for fast fading, respectively.  From the proof of
Lemmas~\ref{thm:mac},~\ref{thm:bc}, and of
Theorem~\ref{thm:achievability} for the fast fading case in
Section\ref{sec:fast}, we see that 
\begin{equation*}
    \bigcap_{i=1}^4 B_{u,w,k}^{(i)} \subset B_{u,w,k}.
\end{equation*}
Due to spatial re-use, multiple relay squarelets will be active in
parallel. Let $\widetilde{H}$ denote the set of channel gains between
active relay squarelets. Using essentially the same arguments as for the
fast fading case (see Lemmas~\ref{thm:mac},~\ref{thm:bc}, and Section
\ref{sec:fast}) and from Markov's inequality, we have
$\Pp(B_{u,w,k}^{(i)}|\widetilde{H}) \geq 7/8$ for all
$i\in\{1,\ldots,4\}$ and hence $\Pp(B_{u,w,k}|\widetilde{H}) \geq 1/2$.

We now argue that the events 
\begin{equation}
    \label{eq:events}
    \Big\{\cap_{i=1}^4 B_{u,w,k}^{(i)}\Big\}_{k=1}^{K_22^{-\ell-1}\gamma(n)}
\end{equation}
are independent conditioned on $\widetilde{H}$, by showing that these
events depend on disjoint sets of channel gains and codebooks.  Assuming
the codebooks are generated new for each communication round, then they
are all independent. Thus we only have to consider the dependence on the
channel gains. Let $U_k$ and $W_k$ be the source and destination nodes
communicating over relay squarelet $A_{u,w,k}$ in round $k$, and let
$V_k$ be the nodes in $A_{u,w,k}$.  Let $\widetilde{U}_k$,
$\widetilde{W}_k$ be the source and destination nodes that are
communicating at the same time as $(u,w)$ due to spatial re-use. Let
$\widetilde{V}_k$ be the relay nodes of $\widetilde{U}_k$ and
$\widetilde{W}_k$. Now, $B_{u,w,k}^{(1)}$ and $B_{u,w,k}^{(2)}$ depend
(for fixed $\widetilde{H}$) on the channel gains between $U_k$ and
$V_k$.  $B_{u,w,k}^{(3)}$ depends on the channel gains between $V_k$ and
$W_k$. $B_{u,w,k}^{(4)}$ depends (again for fixed $\widetilde{H}$) on
the channel gains between $\widetilde{V}_{k}$ and $\widetilde{W}_k$.
Since these sets are disjoint for different $k$ by the orthogonality of
the schedules (see \eqref{eq:disjoint}), conditional independence of the
events in~\eqref{eq:events} follows. 

To summarize, conditioned on the channel gains $\widetilde{H}$ between
active relay squarelets, the random variables $\{\ind_{B_{u,w,k}}\}_k$
are independent and have expected value
$\E(\ind_{B_{u,w,k}}\vert\widetilde{H}) \geq 1/2$. The sum
\begin{equation*}
    \sum_{k=1}^{K_22^{-\ell-1}\gamma(n)}\ind_{B_{u,w,k}}
\end{equation*}
is the number of relay squarelets over which the source-destination pair
$(u,w)$ successfully relays traffic. We now show that with high
probability at least one quarter of these relay squarelets allow
successful transmission. Applying the Chernoff bound yields that 
\begin{align*}
    \Pp \Big( {\textstyle\sum_{k}} \ind_{B_{u,w,k}} < K_22^{-\ell-3}\gamma(n)
    \Big| \widetilde{H} \Big)
    & \leq \Pp\Big( {\textstyle\sum_{k}} \ind_{B_{u,w,k}} 
    < K_22^{-\ell-2}\gamma(n)\Pp(B_{u,w,k}\vert\widetilde{H})
    \Big| \widetilde{H} \Big) \\
    & \leq \exp\big(-2K 2^{-\ell}\gamma(n)\Pp(B_{u,w,k}\vert\widetilde{H})\big) \\
    & \leq \exp\big(-K 2^{-\ell}\gamma(n)\big)
\end{align*}
for some constant $K>0$. Since the right-hand side is the same for all
$\widetilde{H}$, this implies 
\begin{equation*}
    \Pp\Big({\textstyle\sum_{k}} \ind_{B_{u,w,k}} < K_22^{-\ell-3}\gamma(n) \Big)
    \leq \exp\big(-K 2^{-\ell}\gamma(n)\big).
\end{equation*}

In each of the $L(n)$ levels of the hierarchy there are at most $n^2$
source-destination pairs, and hence by the union bound with probability
at least
\begin{equation*}
    1-L(n)n^2\exp\big(-K 2^{-L(n)}\gamma(n)\big),
\end{equation*}
for every source-destination pair on every level of the hierarchy at
least one quarter of its relay squarelets can support the desired rate.  By
the choices of $\gamma(n)$ and $L(n)$ in~\eqref{eq:gammadef}, this
probability is at least 
\begin{align*}
    1-  L(n) n^2\exp\big(-K 2^{-L(n)}\gamma(n)\big) 
    & \geq 1-n^3\exp\Big(-K 2^{-L(n)}2^{\log(n)/2L(n)}\Big) \\
    & \geq 1-\exp\Big(\widetilde{K}2^{\log\log(n)}-K 2^{\frac{1}{2}\log^{1/2+\delta}(n)-\log^{1/2-\delta}(n)}\Big) \\
    & \geq 1-\exp\Big(-2^{\Omega(\log^{1/2+\delta}(n))}\Big) \\
    & \geq 1-o(1)
\end{align*}
as $n\to\infty$, and for some constant $\widetilde{K}$. This proves that
the same order rate as in the fast fading case can be achieved with high
probability for all levels $0\leq\ell<L(n)$.

It remains to argue that the same holds for level $\ell=L(n)$. Note that
since we assume phase fading only, the received signal power is only a
function of distance and not of the fading realization.  Since at level
$L(n)$ we use simple TDMA, this implies that we can always achieve the
same rate at level $L(n)$ as in the fast fading case. 

Hence with probability $1-o(1)$ as $n\to\infty$, we achieve the same
order rate at each level $0\leq \ell\leq L(n)$ as for fast fading,
proving Theorem \ref{thm:achievability} for the slow fading case.

\section{Proof of Theorem \ref{thm:converse}}
\label{sec:converse}

Here, we provide a generalization and sharpening of the converse in
\cite{ozg}. Most of the arguments follow \cite[Theorem 5.2]{ozg}.
We start by proving a lemma upper bounding the MIMO capacity.

Consider two subsets $S_1,S_2\subset V(n)$ such that $S_1\cap S_2 =
\emptyset$.  Assume we allow the nodes within $S_1$ and $S_2$ to
cooperate without any restriction. The maximum achievable sum rate
between the nodes in $S_1$ and $S_2$ is given by the MIMO capacity
$C(S_1,S_2)$ between them. The next lemma upper bounds $C(S_1,S_2)$ in
terms of the node distances between the two sets and the
\emph{normalized channel gains}
\begin{equation*}
    \tilde{h}_{u,v} 
    \defeq \frac{h_{u,v}}{\sqrt{\sum_{\tilde{v}\in S_2}r_{u,\tilde{v}}^{-\alpha}}}.
\end{equation*}

\begin{lemma}
    \label{thm:mimo}
    Under either fast or slow fading, for every $\alpha > 2$,
    $S_1,S_2\subset V(n)$ with $S_1\cap S_2 = \emptyset$, we have
    \begin{equation*}
        C(S_1,S_2) 
        \leq 4\bigg(\max\bigg\{1,\max_{v\in S_2}\sum_{u\in S_1}\lvert\tilde{h}_{u,v}\rvert^2\bigg\}\bigg)
        \sum_{u\in S_1} \sum_{v\in S_2} r_{u,v}^{-\alpha}.
    \end{equation*}
\end{lemma}
\begin{proof}
    Let
    \begin{align*}
        \bm{H} & \defeq \{h_{u,v}\}_{u\in S_1,v\in S_2}, \\
        \widetilde{\bm{H}} & \defeq \{\tilde{h}_{u,v}\}_{u\in S_1,v\in S_2}, 
    \end{align*}
    be the matrix of (normalized) channel gains between the nodes in $S_1$ and $S_2$.
    Consider first fast fading. Under this assumption, we have 
    \begin{equation*}
        C(S_1,S_2) \\
        \defeq 
        \max_{\substack{\bm{Q}(\bm{H})\geq 0: \\ \E(q_{u,u})\leq 1\ \forall u\in S_1}}
        \E\bigg( \log \det\big(\bm{I}+\bm{H}^\dagger\bm{Q}(\bm{H}) \bm{H}\big)\bigg).
    \end{equation*}
    Define 
    \begin{equation*}
        P_{S_1,S_2} \defeq \sum_{u\in S_1} \sum_{v\in S_2}r_{u,v}^{-\alpha}
    \end{equation*}
    as the total received power in $S_2$ from $S_1$, and set
    \begin{equation*}
        P_{u,S_2}\defeq P_{\{u\},S_2}
    \end{equation*}
    with slight abuse of notation. Then
    \begin{align}
        \label{eq:mimo1}
        C(S_1,S_2) 
        & = \max_{\substack{\bm{Q}(\bm{H})\geq 0: \\ \E(q_{u,u})\leq P_{u,S_2}\forall u\in S_1}}
        \E\bigg( \log \det\big(\bm{I}+\widetilde{\bm{H}}^\dagger\bm{Q}(\bm{H})\widetilde{\bm{H}}\big)\bigg) \nonumber\\
        & \leq\max_{\substack{\bm{Q}(\bm{H})\geq 0: \\ \E(\tr\bm{Q}(\bm{H}))\leq P_{S_1,S_2}}}
        \E\bigg( \log \det\big(\bm{I}+\widetilde{\bm{H}}^\dagger\bm{Q}(\bm{H})\widetilde{\bm{H}}\big)\bigg).
    \end{align}

    Define the event
    \begin{equation*}
        B \defeq \big\{ \lVert\widetilde{\bm{H}}\rVert^2 > b\big\}
    \end{equation*}
    for some $b$ and where $\lVert\widetilde{\bm{H}}\rVert$ denotes the
    largest singular value of $\widetilde{\bm{H}}$. In words, $B$ is the
    event that the channel gains between $S_1$ and $S_2$ are ``good''.
    We argue that, for appropriately chosen $b$, the event $B$ has
    probability zero (i.e., the channel can not be too ``good''). By Markov's
    inequality
    \begin{equation}
        \label{eq:mimo2}
        \Pp(B) \leq b^{-m}\E(\lVert\widetilde{\bm{H}}\rVert^{2m}),
    \end{equation}
    for any $m$. We continue by upper bounding
    $\E(\lVert\widetilde{\bm{H}}\rVert^{2m})$. We have
    \begin{equation*}
        \lVert\widetilde{\bm{H}}\rVert^{2k}
        \leq \tr{\big((\widetilde{\bm{H}}\widetilde{\bm{H}}^\dagger)^k\big)}
    \end{equation*}
    for any $k$, and hence
    \begin{equation}
        \label{eq:mimo3}
        \E(\lVert\widetilde{\bm{H}}\rVert^{2m})
        \leq \E\Big( \big( \tr{\big((\widetilde{\bm{H}}\widetilde{\bm{H}}^\dagger)^k\big)} \big)^{m/k} \Big).
    \end{equation}
    Now, for any $k \geq m$, we have by Jensen's inequality
    \begin{equation}
        \label{eq:mimo4}
        \E\Big( \big( \tr{\big((\widetilde{\bm{H}}\widetilde{\bm{H}}^\dagger)^k\big)} \big)^{m/k} \Big)
        \leq \Big( \E \tr{\big((\widetilde{\bm{H}}\widetilde{\bm{H}}^\dagger)^k} \big) \Big)^{m/k}.
    \end{equation}
    Combining \eqref{eq:mimo2}, \eqref{eq:mimo3}, and \eqref{eq:mimo4}
    yields
    \begin{equation}
        \label{eq:mimo5}
        \Pp(B) 
        \leq b^{-m}\Big( \E \tr{\big((\widetilde{\bm{H}}\widetilde{\bm{H}}^\dagger)^k} \big) \Big)^{m/k}
    \end{equation}
    for any $k \geq m$.

    Now, the arguments in \cite[Lemma 5.3]{ozg} show that
    \begin{equation*}
        \E\big( \tr{\big((\widetilde{\bm{H}}\widetilde{\bm{H}}^\dagger)^k\big)} \big)
        \leq t_k n \bigg(\max\bigg\{1,
        \max_{v\in S_2}\sum_{u\in S_1}\lvert\tilde{h}_{u,v}\rvert^2\bigg\}\bigg)^k,
    \end{equation*}
    where $t_{k}$ is the $k$-th Catalan number. Combining with
    \eqref{eq:mimo5}, this yields
    \begin{equation*}
        \Pp(B) 
        \leq \bigg(b^{-1}t_k^{1/k} n^{1/k} 
        \Big(\max\Big\{1,\max_{v\in S_2}\sum_{u\in S_1}\lvert\tilde{h}_{u,v}\rvert^2\Big\}\Big)
        \bigg)^m.
    \end{equation*}
    Taking the limit as $k\to\infty$ and using that $t_k^{1/k}\to 4$
    yields
    \begin{equation*}
        \Pp(B) 
        \leq \bigg(b^{-1}4 
        \Big(\max\Big\{1,\max_{v\in S_2}\sum_{u\in S_1}\lvert\tilde{h}_{u,v}\rvert^2\Big\}\Big)
        \bigg)^m.
    \end{equation*}
    Assume
    \begin{equation}
        \label{eq:mimo6}
        b > 4 \Big(\max\Big\{1,\max_{v\in S_2}\sum_{u\in S_1}\lvert\tilde{h}_{u,v}\rvert^2\Big\}\Big), 
    \end{equation}
    then taking the limit as $m\to\infty$ shows that
    \begin{equation*}
        \Pp(B) = 0.
    \end{equation*}
 
    Using this, we can upper bound \eqref{eq:mimo1} as
    \begin{align*}
        C (S_1, S_2) 
        & \leq\max_{\substack{\bm{Q}(\bm{H})\geq 0: \\ \E(\tr\bm{Q}(\bm{H}))\leq P_{S_1,S_2}}}
        \E\Big( \tr\Big(\widetilde{\bm{H}}^\dagger\bm{Q}(\bm{H})\widetilde{\bm{H}}\Big)\Big) \\
        & = \max_{\substack{\bm{Q}(\bm{H})\geq 0: \\ \E(\tr\bm{Q}(\bm{H}))\leq P_{S_1,S_2}}}
        \E\Big( \ind_{B^c} \tr\Big(\widetilde{\bm{H}}^\dagger\bm{Q}(\bm{H})\widetilde{\bm{H}}\Big)\Big) \\
        & \leq\max_{\substack{\bm{Q}(\bm{H})\geq 0: \\ \E(\tr\bm{Q}(\bm{H}))\leq P_{S_1,S_2}}}
        \E\Big( \ind_{B^c} \lVert\widetilde{\bm{H}}\rVert^2\tr\bm{Q}(\bm{H})\Big) \\
        & \leq b P_{S_1,S_2}.
    \end{align*}
    Since this is true for all $b$ satisfying \eqref{eq:mimo6}, we
    obtain the lemma for the fast fading case.

    Under slow fading
    \begin{equation*}
        C(S_1,S_2)
        \defeq \max_{\substack{\bm{Q}\geq 0: \\ q_{u,u}\leq P\ \forall u\in S_1}}
        \log \det\big(\bm{I}+\bm{H}^\dagger\bm{Q} \bm{H}\big),
    \end{equation*}
    and the lemma can be obtained by the same steps.
\end{proof}

We now proceed to the proof of Theorem \ref{thm:converse}.  Consider a
vertical cut dividing the network into two parts.  By the
minimum-separation requirement, an area of size $o(n)$ can contain at
most $o(n)$ nodes, and hence we can find a cut such that each part is of
size $\Theta(n)$ and contains $\Theta(n)$ nodes. Call the left part of
the cut $S$. Since there are $\Theta(n)$ nodes in $S$ and in $S^c$,
there are $\Theta(n)$ sources in $S$ with their destination in $S^c$
with probability $1-o(1)$.  For technical reasons we add a node inside
each square in $V(n)$ of the form $[i d,(i+1)d]\times[jd,(j+1)d]$ for
some $i,j\in\mbb{N}$, where $d\defeq\sqrt{2\log(n)}$. These additional
nodes have no traffic demands on their own, and simply help with the
transmission. This can clearly only increase achievable rates. Moreover,
this increases the number of nodes in $V$ by less than a factor $2$. We
now show that
\begin{equation}
    \label{eq:converse1}
    C(S,S^c) = O\big(\log^6(n) n^{2-\alpha/2}\big),
\end{equation}
and hence by the cut-set bound, and since there are $\Theta(n)$ sources
in $S$ with their destination in $S^c$, we have
\begin{equation*}
    \rho^*(n) = O\big(\log^6(n)n^{1-\alpha/2}\big).
\end{equation*}

We prove \eqref{eq:converse1} using Lemma \ref{thm:mimo}. To this end,
we need to upper bound
\begin{equation*}
    \max_{v\in S^c}\sum_{u\in S}\lvert\tilde{h}_{u,v}\rvert^2.
\end{equation*}
The proof of \cite[Lemma 5.3]{ozg} shows that if
\begin{enumerate}
    \item there are less than $\log(n)$ nodes inside
        $[i,i+1]\times[j,j+1]$ for any $i,j\in\{0,\ldots,\sqrt{n}-1\}$,
    \item there is at least one node inside $[i d,(i+1)d]\times[jd,(j+1)d]$
        for any $i,j$, where $d\defeq \sqrt{2\log n}$,
\end{enumerate}
then
\begin{equation}
    \label{eq:converse2}
    \max_{v\in S^c}\sum_{u\in S}\lvert\tilde{h}_{u,v}\rvert^2
    \leq K \log^3(n),
\end{equation}
and for $\alpha\in(2,3]$
\begin{equation}
    \label{eq:converse3}
    \sum_{u\in S} \sum_{v\in S^c} r_{u,v}^{-\alpha}
    \leq \widetilde{K}\log^3(n)n^{2-\alpha/2},
\end{equation}
for constants $K,\widetilde{K}$. For arbitrary node placement with
minimum separation, the first requirement is satisfied for $n$ large
enough, since only a constant number of nodes can be contained in each
area of constant size. By our addition of nodes into $V(n)$ described
above, the second condition is also satisfied. Using Lemma
\ref{thm:mimo} with \eqref{eq:converse2} and \eqref{eq:converse3} yields
\eqref{eq:converse1}, concluding the proof of Theorem
\ref{thm:converse}.

\section{Proof of Theorem~\ref{thm:adversary}}
\label{sec:adversary}

Consider a node placement with $n/2$ nodes located uniformly on
$[0,\sqrt{n}/4]\times[0,\sqrt{n}]$ and $n/2$ nodes located on
$[\sqrt{n}/2,\sqrt{n}]\times[0,\sqrt{n}]$ with minimum separation
$r_{\min}=1/2$.  A random traffic matrix $\lambdauc(n)$ is such that at least
$n/4$ communication pairs have their sources in the left cluster and
destinations in the right cluster with probability $1-o(1)$.  Assume we
are dealing with such a $\lambdauc(n)$ in the following.

In this setup, with multi-hop at least one hop has to cross the gap
between the left and the right cluster. Thus, even without any
interference from other nodes, we can obtain at most
\begin{equation*}
    \rho^{\tsf{MH}}(n) \leq 4^\alpha n^{-\alpha/2}.
\end{equation*}

Moreover, considering a cut between the two clusters (say, $S$ and
$S^c$), and applying Lemma \ref{thm:mimo} yields that 
\begin{equation}
    \label{eq:adversary1}
    \rho^*(n) \leq 16n^{-1}
    \bigg(\max\bigg\{1,\max_{v\in S^c}\sum_{u\in S}\lvert\tilde{h}_{u,v}\rvert^2\bigg\}\bigg)
        \sum_{u\in S} \sum_{v\in S^c} r_{u,v}^{-\alpha}.
\end{equation}
Now note that for any $u\in S$, $v\in S^c$, we have
\begin{equation*}
    \frac{1}{4}\sqrt{n} 
    \leq r_{u,v}
    \leq 2\sqrt{n}.
\end{equation*}
Hence
\begin{equation*}
    \sum_{u\in S}\lvert\tilde{h}_{u,v}\rvert^2
    = \sum_{u\in S} \frac{r_{u,v}^{-\alpha}}{\sum_{\tilde{v}\in S^c} r_{u,\tilde{v}}^{-\alpha}}
    \leq 2^{3\alpha},
\end{equation*}
and
\begin{equation*}
    \sum_{u\in S}\sum_{v\in S^c}r_{u,v}^{-\alpha}
    \leq 4^{\alpha-1}n^{2-\alpha/2}.
\end{equation*}
Combining this with \eqref{eq:adversary1} yields
\begin{equation*}
     \rho^*(n) \leq 2^{2+5\alpha}n^{1-\alpha/2}
\end{equation*}
for all $\alpha >2$.

\section{Proof of Theorem \ref{thm:achievabilityx}}
\label{sec:multihop}

We construct a cooperative multi-hop communication scheme and lower
bound the per-node rate $\rho^{\tsf{CMH}}(n)$ it achieves. We use the
hierarchical relaying scheme as building block. Assume the node
placement $V(n)$ is $\mu$-regular at resolution $d(n)$ for all $n\geq
1$. We show that this implies that we can achieve a per-node rate of at
least $d^{3-\alpha}(n)n^{-1/2-\beta(n)}$ as $n\to\infty$. Taking the
smallest such $d(n)$ then yields the result. 

We consider three cases for the value of $d(n)$ (namely, $d(n) =
\Theta(\sqrt{n})$, $d(n) \geq n^{o(1)}$, and $d(n) \leq n^{o(1)}$).
First, if $d(n) = \Theta(\sqrt{n})$ as $n\to\infty$ then the result follows
directly from Theorem~\ref{thm:achievability}. Considering a subsequence
if necessary, we can therefore assume without loss of generality that
$d(n) = o(\sqrt{n})$ in the following.

Second, consider $d(n)$ satisfying
\begin{equation}
    \label{eq:hassumption}
    d(n) \geq n^{\frac{1}{2+\alpha}\log^{\delta-1/2}(n)}.  
\end{equation}
Divide $A(n)$ into squares of sidelength $d(n)$.  Since
$d(n) = o(\sqrt{n})$, the number of such squares grows unbounded as
$n\to\infty$. We now show that we can use multi-hop communication with a
hop length of $d(n)$ where each hops is implemented by squares
cooperatively sending information to a neighboring square. In other words,
we perform cooperative communication at local scale $d(n)$ and multi-hop
communication at global scale $\sqrt{n}$. 

Since $V(n)$ is $\mu$-regular at resolution $d(n)$, each such square
contains at least $\mu d^2(n)$ nodes.  Pick the top left most square and
construct the square of sidelength $2d(n)$ consisting of it together
with its $3$ neighbors.  Continue in the same fashion, partitioning all
of $A(n)$ into squares of sidelength $2d(n)$. Note that each such bigger
square contains at least $4\mu d^2(n)$ nodes by the definition of
$d(n)$. We assume this worst case in the following.  Partition $A(n)$
into $4$ subsets of those bigger squares such that within each such
subset each square is at distance at least $2d(n)$ from any other square
(see Figure~\ref{fig:backbone}). We time share between those $4$
subsets.  Consider in the following one such subset. For every bigger
square, we construct two permutation traffic matrices $\lambdauc_1(4\mu d^2(n))$ and
$\lambdauc_2(4\mu d^2(n))$. In $\lambdauc_1$ the nodes in the top two squares have as
destinations the nodes in the bottom two squares and the nodes in the
bottom two squares have as destinations the nodes in the top two squares
(see Figure~\ref{fig:backbone}). Similarly, $\lambdauc_2$ contains communication
pairs between left and right squares.  We time share between $\lambdauc_1$ and
$\lambdauc_2$.
\begin{figure}[!ht]
    \begin{center}
        \input{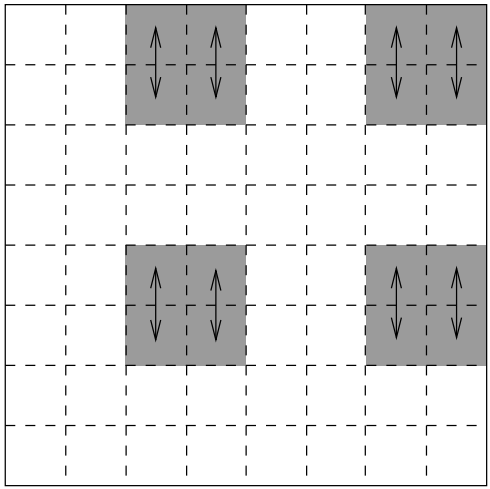}
    \end{center}
    \caption{Sketch of the construction of the cooperative multi-hop
    scheme in the proof of Theorem~\ref{thm:achievabilityx}. The dashed
    squares have sidelength $d(n)$. The gray area is one of the $4$
    subsets of bigger squares that communicate simultaneously. The
    arrows indicate the traffic matrix $\lambdauc_1$.}
    \label{fig:backbone}
\end{figure}

Communication according to $\lambdauc_i$ within bigger squares in the same
subset occurs simultaneously. We are going to use hierarchical relaying
within each bigger square. This is possible since each such square
contains at least $4\mu d^2(n)$ nodes. We have to show that the
additional interference from bigger squares in the same subset is such
that Theorem~\ref{thm:achievability} still applies.  In particular, we
need to show that the interference has bounded power, say $K$. Using the
same arguments as in the proof of Theorem~\ref{thm:achievability} in
Section \ref{sec:achievability} yields that this is indeed the case (the
interference from other bigger squares here behaves the same way as the
interference due to spatial re-use from other active relay squarelets
there). With this, we are now dealing with a hierarchical relaying
scheme with area $4d^2(n)$, $4\mu d^2(n)$ nodes, and additive noise with
power $1+K$.  Both the lower number of nodes and the higher noise power
will decrease the achievable per-node rate by only some constant factor,
and hence Theorem~\ref{thm:achievability} shows that under 
fast fading we can achieve a per-node rate of at least 
\begin{equation*}
    b_1\big(d^2(n)\big)(d^2(n))^{1-\alpha/2} \geq b_1(n)d^{2-\alpha}(n),
\end{equation*}
as $n\to\infty$, where 
\begin{equation*}
    b_1(n) \geq n^{-O\big(\log^{\delta-1/2}(n)\big)}.
\end{equation*}
Moreover, the same rate is
achievable under slow fading with probability $1- b_2(d^2(n))$, where
\begin{equation*}
    b_2(n) \leq \exp\Big(-2^{\Omega\big(\log^{1/2+\delta}(n)\big)}\Big).
\end{equation*}
The setup is the same for all bigger squares within each of the $4$ subsets. 

We now ``shift'' the way we defined the bigger squares by $d(n)$ to
the right and to the bottom. With this, each new bigger square
intersects with $4$ bigger squares as defined before. We use the same
communication scheme within these new bigger squares and time share
between the two ways of defining bigger squares. 

Construct now a
graph where each vertex corresponds to a square of sidelength $d(n)$
and where two vertices are connected by an edge if they are adjacent
in either the same old or new bigger square. This graph is depicted in
Figure~\ref{fig:backbone2} in Section \ref{sec:schemes_multihop}.

With the above construction, we can communicate along each edge of this
graph simultaneously at a per-node rate of
\begin{equation*}
    \frac{b_1(n)}{16}d^{2-\alpha}(n)
\end{equation*}
in the fast fading case.  In the slow fading case, this statement holds
with probability at least
\begin{align*}
    1-\frac{n}{d^2(n)}b_2(d^2(n))
    & = 1-\frac{n}{d^2(n)}\exp\Big(-2^{\Omega\big(\log^{1/2+\delta}(d^2(n))\big)}\Big) \\
    & \geq 1-\exp\Big(K'2^{\log\log(n)}-2^{\widetilde{K}\log^{1/2+\delta}(d(n))}\Big)
\end{align*}
for constants $K',\widetilde{K}$. By assumption
\eqref{eq:hassumption}, 
\begin{equation*}
    \log^{1/2+\delta}\big(d(n)\big)
    \geq \Big(\frac{1}{2+\alpha}\log^{1/2+\delta}(n)\Big)^{1/2+\delta},
\end{equation*}
and hence
\begin{equation*}
    1-\frac{n}{d^2(n)}b_2(d^2(n)) \geq 1-o(1)
\end{equation*}
as $n\to\infty$, showing that with high probability we achieve the same
order rate under slow fading as under fast fading.

The communication graph constructed forms a grid with $n/d^2(n)$ nodes.
Using that each bigger square can contain at most $K_1 d^2(n)$ nodes by
the minimum-separation requirement, standard arguments for routing over
grid graphs (see \cite{kul}) show that in the fast fading case we can
achieve a per-node rate of
\begin{equation*}
    \rho^{\tsf{CMH}}(n) 
    \geq \tilde{b}(n)d^{2-\alpha}(n)\frac{d(n)}{\sqrt{n}}
    \geq \tilde{b}(n)d^{3-\alpha}(n)n^{-1/2},
\end{equation*}
where 
\begin{equation*}
    \tilde{b}(n)= n^{-O\big(\log^{\delta-1/2}(n)\big)}. 
\end{equation*}
Moreover, the same statement holds in the slow fading case with
probability $1-o(1)$.

Finally, consider $d(n)$ such that
\begin{equation}
    \label{eq:hassumption2}
    d(n) \leq n^{\frac{1}{2+\alpha}\log^{\delta-1/2}(n)}.  
\end{equation}
Construct the same communication graph as before, but this time we use
simple multi-hop communication between adjacent squares of sidelength
$d(n)$. By time sharing between the at most $K_1 d^2(n)$ nodes in each
square, and since we communicate over a distance of at most $3d(n)$, we
achieve under either fast of slow fading a per-node rate between the
squares of at least
\begin{equation*}
    K'' d^{-2-\alpha}(n)
    \geq K'' n^{-\log^{\delta-1/2}(n)}
\end{equation*}
for some constant $K''$, and where we have used \eqref{eq:hassumption2}.
Using the analysis of grid graphs as before, we can achieve a per-node
rate of at least
\begin{equation*}
    \rho^{\tsf{CMH}}(n) 
    \geq K'' n^{-\log^{\delta-1/2}(n)}\frac{d(n)}{\sqrt{n}} 
    \geq \tilde{b}(n)d^{3-\alpha}(n)n^{-1/2},
\end{equation*}
for either the fast or slow fading case.

\section{Proof of Theorem \ref{thm:adversaryx}}
\label{sec:adversaryx}

Consider $V(n)$ with $n/2$ nodes located uniformly on
$[0,(\sqrt{n}-d^*(n))/2]\times[0,\sqrt{n}]$ and $n/2$ nodes located
uniformly on $[\sqrt{n}/2,\sqrt{n}]\times[0,\sqrt{n}]$ such that
$r_{\min}=1/2$. This node placement is $1/2$-regular at resolution
$d^*(n)$. A random traffic matrix $\lambdauc(n)$ is such that $\Theta(n)$
communication pairs have their sources in the left cluster and
destinations in the right cluster with probability $1-o(1)$. Assume we
are dealing with such a $\lambdauc(n)$ in the following. 

Considering a cut between the two clusters and applying Lemma
\ref{thm:mimo} (slightly adapting the arguments in Section
\ref{sec:converse}), yields that
\begin{equation*}
    \rho^*(n) 
    = O\big(\log^6(n) {d^*}^{3-\alpha}(n)n^{-1/2}\big)
\end{equation*}
for $\alpha>3$.

\section{Discussion}
\label{sec:discussion}

We briefly discuss several aspects of the proposed hierarchical
relaying scheme. Section \ref{sec:csi} comments on the full CSI
assumption and Section \ref{sec:bursty} on the use of bursty communication.
Sections \ref{sec:dense} and \ref{sec:dmin} outline how the results
obtained here can be extended to the case of dense networks and networks
without minimum separation between nodes. Section \ref{sec:comparison}
compares our hierarchical relaying scheme to the hierarchical
cooperation scheme presented in \cite{ozg}.

\subsection{Full CSI Assumption}
\label{sec:csi}

Throughout our analysis, we have made a full CSI assumption. In other
words, we assumed that the phase shifts $\{\theta_{u,v}[t]\}_{u,v}$ are
available at time $t$ at all nodes in the network.  As this assumption
is quite strong, it is worth commenting on. First, we make the full CSI
assumption in all the converse results in this paper. This implies that
all the converses also hold under weaker assumptions on the CSI, and
hence are valid as well under a wide variety of more realistic
assumptions on the availability of side information. Second, all
achievability results can be shown to hold under weaker assumptions on
the availability of CSI. In fact, in all cases, a $2$-bit quantization
of the channel state $\{\theta_{u,v}[t]\}_{u,v}$ available at all nodes 
at time $t$ is sufficient to obtain the same scaling behavior. This
follows by an argument similar to the one used in the analysis of the BC
phase in Section \ref{sec:bc}, where it is shown that beamforming using
a quantized channel state results only in a constant factor rate loss.

\subsection{Burstiness of Hierarchical Relaying Scheme}
\label{sec:bursty}

The hierarchical relaying scheme presented here is bursty in the sense
that nodes communicate at high power during a small fraction of time.
This leads to high peak-to-average power ratio, which is undesirable
in practice. We chose burstiness in the time domain to simplify the
exposition. The same bursty behavior could be achieved in a more
practical manner by using CDMA with several orthogonal signatures or by using
OFDM with many sub-carriers. Each approach leads to many parallel
channels out of which only few are used with higher power. This avoids
the issue of high peak-to-average power ratio in the time domain.

\subsection{Dense Networks}
\label{sec:dense}

Throughout this paper, we have only considered \emph{extended} networks,
i.e, $n$ nodes placed on a square region of area $n$ with a minimum
separation of $r_{u,v}\geq r_{\min}$. The results can, however, be
recast for \emph{dense} networks, where $n$ nodes are arbitrarily placed
on a square region of unit area with a minimum separation of
$r_{u,v}\geq r_{\min}/\sqrt{n}$. It suffices to notice that by rescaling
power by a factor $n^{-\alpha/2}$ a dense network can essentially be
transformed into an extended network with path-loss exponent $\alpha$
(see also \cite{ozg}).  Hence the same result for dense networks can be
obtained from the result for extended networks by considering the limit
$\alpha\to 2$. Applying this to Theorem \ref{thm:achievability}, yields
a linear per-node rate scaling of the hierarchical relaying scheme.

\subsection{Minimum-Separation Requirement}
\label{sec:dmin}

The minimum-separation requirement $r_{\min}\in(0,1)$ on the node
placement is sufficient but not necessary for
Theorem~\ref{thm:achievability} to hold. A weaker sufficient condition
is that a constant fraction of squarelets are dense, as shown in
Lemma~\ref{thm:structure} to be a consequence of the minimum-separation
requirement. It is straightforward to show that this weaker condition is
satisfied with high probability for nodes placed uniformly at random on
$[0,\sqrt{n}]^2$. This yields a different proof of Theorem 5.1
in~\cite{ozg}.

\subsection{Comparison with~\cite{ozg}}
\label{sec:comparison}

Both, the hierarchical relaying scheme presented here and the
hierarchical scheme presented in~\cite{ozg}, share that they use virtual
multiple-antenna communication and a hierarchical architecture to
achieve essentially global cooperation in the network.  The schemes
differ, however, in several key aspects, which we point out here.

First, we note that we obtain a slightly better scaling law. Namely
\begin{equation*}
    b_1(n)n^{1-\alpha/2}
    \leq \rho^*(n)
    \leq b_2(n)n^{1-\alpha/2}
\end{equation*}
with
\begin{align*}
    b_1(n) & \geq n^{-O\big(\log^{\delta-1/2}(n)\big)}, \\
    b_2(n) & = O\big(\log^6(n)\big),
\end{align*}
for any $\delta\in(0,1/2)$ obtained here, compared to
\begin{equation*}
    \tilde{b}_1(n)n^{1-\alpha/2}
    \leq \rho^*(n)
    \leq \tilde{b}_2(n)n^{1-\alpha/2}
\end{equation*}
with
\begin{align*}
    \tilde{b}_1(n) & = \Omega\big(n^{-\varepsilon}\big), \\
    \tilde{b}_2(n) & = O\big(n^{\varepsilon}\big),
\end{align*}
for any $\varepsilon>0$ in \cite{ozg}.  For the lower bound (i.e.,
achievability), this is because the hierarchy here is not of fixed depth
$L$ as in~\cite{ozg}, but rather of depth $L(n)=\log^{1/2-\delta}(n)$
(for some constant $\delta\in(0,1/2)$), i.e., changing with $n$. For the
upper bound (i.e., converse), this is due to a sharpening of the
arguments in \cite{ozg}.

Second, note that the multi-user decoding at the relay squarelets during
the MAC phase and the multi-user encoding during the BC phase are very
simple in our setup. In fact, using matched filter receivers and
transmit beamforming, we convert the multi-user encoding and decoding
problems into several single-user decoding and encoding problems. This
differs from the approach in~\cite{ozg}, in which joint decoding of a
number of users on the order of the network size is performed. Our
results thus imply that these simpler transmitter and receiver
structures provide the same scaling as the more complicated joint
decoding in~\cite{ozg}. We note that the scheme proposed in \cite{ozg}
can be modified to also use matched filter receivers as suggested here.

Third, and probably most important, the schemes differ in how they
achieve the throughput gain from using multiple antennas.
In~\cite{ozg}, the nodes are located almost regularly with high
probability. This allowed the use of a scheme in which a source
squarelet directly communicates with a destination squarelet. In other
words, the multiple-antenna gain comes from setting up a virtual MIMO
channel between the source and the destination. In our setup, the
arbitrary location of nodes prevents such an approach.  Instead, we use
that at least some fixed fraction of squarelets is almost regular (we
called them dense squarelets). Source-destination pairs relay their
traffic over such a dense squarelet. In other words, the
multiple-antenna gain comes from setting up a virtual multiple-antenna
MAC and BC. Thus, the hierarchical relaying scheme presented here shows
that considerably less structure on the node locations than assumed
in~\cite{ozg} suffices to achieve a multiple-antenna gain essentially on
the order of the network size. Note also that the additional degree of
freedom offered by the choice of relay squarelet for a given
source-destination pair makes it possible to extend the result to hold
also for slow fading channels.

\section{Conclusions}
\label{sec:conclusions}

We considered the problem of the scaling of achievable rates in
arbitrary extended wireless networks. We generalized the hierarchical
cooperative communication scheme presented in~\cite{ozg} for a fast
fading channel model and with random node placements. We proposed a
different hierarchical cooperative communication scheme, which also
works for arbitrary node placement (with a minimum-separation
requirement) and for either fast or slow fading. 

For small path-loss exponent $\alpha\in(2,3]$, we showed that our scheme
is order optimal and achieves the same rate irrespective of the node
placement. In particular, this rate is equal to the one achievable under
random node placement.  In other words, the regularity of the node
placement has no impact on achievable rates for small path-loss
exponent.

The situation is, however, quite different for large path-loss exponent
$\alpha >3$. We argued that in this regime the regularity of the node
placement directly impacts the scaling of achievable rates. We then
presented a  cooperative communication scheme that smoothly
``interpolates'' between multi-hop and hierarchical cooperative
communication depending on the regularity of the node placement. We
showed that this scheme is order optimal for all $\alpha>3$ under
adversarial node placement with regularity constraint.  This contrasts
with the situation for more regular networks (like the ones obtained
with high probability through random node placement), in which multi-hop
communication is order optimal for all $\alpha > 3$.  Thus, for less
regular networks, the use of more complicated cooperative communication
schemes can be necessary for optimal operation of the network.

\section{Acknowledgments}

The authors would like to thank the anonymous reviewers and the Associate
Editor Gerhard Kramer for their comments. We would also like to
acknowledge helpful discussions with Olivier L{\'e}v{\^e}que, Ayfer
{\"O}zg{\"u}r, and Greg Wornell.

\bibliography{arbitrary}

\end{document}